\documentclass[12pt,reqno]{amsart}

\usepackage[samelinetheorem,notitle]{maherart}
\usepackage{amsfonts,enumerate,bm,bbm,paralist,nicefrac,comment,float,thmtools,thm-restate}

\makeatletter
\newcommand\fs@boxedtop
  {\fs@boxed
   \def\@fs@mid{\vspace\abovecaptionskip\relax}%
   \let\@fs@iftopcapt\iftrue
  }
\makeatother
\floatstyle{boxedtop}
\floatname{framedbox}{Procedure}
\newfloat{framedbox}{tbp}{lob}

\hypersetup{
    pdftitle =
        {},
    pdfauthor =
        {},
    pdfsubject=
        {}
}
\renewcommand{\R}{\mathbb{R}}

\usepackage{xparse}

\DeclareDocumentCommand\Pr{ m g }{%
    \ensuremath{   \IfNoValueTF {#2}
      {\mathbb{\mathbf{p}}\left[{#1}\right]}
      {\mathbb{\mathbf{p}}\left[{#1}\middle\vert{#2}\right]}%
    }
}
\DeclareDocumentCommand\E{ m g }{%
    \ensuremath{   \IfNoValueTF {#2}
      {\mathbb{E}\left[{#1}\right]}
      {\mathbb{E}\left[{#1}\middle\vert{#2}\right]}%
    }
}

\newcommand{\argmin}{\operatornamewithlimits{argmin}}

\renewcommand{\Re}{\R}
\newif\ifdraft
\drafttrue

\newcommand{\mc}[1]{\ifdraft
\textcolor{green!80!black}{[Matteo: #1]} \fi}
\newcommand{\mn}[1]{\ifdraft
\textcolor{blue}{[Mingzi: #1]} \fi}
\newcommand{\Xomit}[1]{}
\usepackage{relsize}
\usepackage{natbib}
\usepackage{float}

\renewcommand{\paragraph}[1]{\bigskip \noindent \textbf{#1}}
\renewcommand{\subparagraph}[1]{\medskip \noindent \emph{#1}}

\begin{document}
\onehalfspacing
\raggedbottom
\title{Signaling Design}

\author[Camboni]{Matteo Camboni$^\text{a}$}
\address{$^\text{\MakeLowercase{a}}$Department of Economics, University of Wisconsin-Madison%
\\
\href{mailto:camboni@wisc.edu}{camboni@wisc.edu}}

\author[Niu]{Mingzi Niu$^\text{b}$}
\address{$^\text{\MakeLowercase{b}}$The Hebrew University Business School
\\
\href{mailto:mingzi.niu@mail.huji.ac.il}{mingzi.niu@mail.huji.ac.il}}

\author[Pai]{Mallesh M. Pai$^\text{c}$}
\address{$^\text{\MakeLowercase{c}}$Department of Economics, Rice University%
\\
\href{mailto:mallesh.pai@rice.edu}{mallesh.pai@rice.edu}}

\author[Vohra]{Rakesh Vohra$^\text{d}$}
\address{$^\text{\MakeLowercase{d}}$Department of Economics, University of Pennsylvania%
\\
\href{mailto:rvohra@seas.upenn.edu}{rvohra@seas.upenn.edu}}

\address{\today}
\thanks{
We thank Ben Brooks, Francesc Dilme, Alessandro Pavan, Maher Said, and Mark Whitmeyer for their helpful comments and discussions. We are grateful to Jorge Zazueta for helpful research assistance---all errors remain our own. Niu gratefully acknowledges support from Azrieli Foundation; Pai gratefully acknowledges support from the NSF (CCF-1763349).}

\begin{abstract}
We revisit the classic job-market signaling model of \cite{spence1973job}, introducing profit-seeking schools as intermediaries that design the mapping from candidates' efforts to job-market signals. Each school commits to an attendance fee and a monitoring policy. We show that, in equilibrium, a monopolist school captures the entire social surplus by committing to low information signals and charging fees that extract students' surplus from being hired. In contrast, competition shifts surplus to students, with schools vying to attract high-ability students, enabling them to distinguish themselves from their lower-ability peers. However, this increased signal informativeness leads to more wasteful effort in equilibrium, contrasting with the usual argument that competition enhances social efficiency. This result may be reversed if schools face binding fee caps or students are credit-constrained.  

\end{abstract}


\maketitle

\thispagestyle{empty}

\newpage

\pagenumbering{arabic} 
\section{Introduction}\label{sec:intro}
\cite{spence1973job}'s signaling model is the main alternative to the human capital model in understanding education's role in job markets. Even if a job candidate acquired no skills in the groves of academe, the ordeal is still useful if graduation correlates with traits predictive of job performance. The idea is straightforward: in a competitive labor market, candidates with privately known abilities can engage in costly signaling, i.e., undertaking actions that are more burdensome for those of lower ability.  While this enables high-ability individuals to secure higher wages, intense effort into signaling may be socially wasteful. Such inefficiency has prompted one critic of higher education to assert that ``taxpayers are mostly fueling a futile arms race.''\footnote{
\url{https://www.latimes.com/opinion/op-ed/la-oe-caplan-education-credentials-20180211-story.html}}

Indeed, the last twenty years have seen an explosion of interest in cheap, fast, and accessible ways for workers to reduce employers' uncertainty about their abilities. These take the form of alternative credentials such as badges and micro-credentials, e.g., privately administered skill certificates awarded based on online tests. Their most attractive feature is that the cost of providing these certificates is a fraction of that required to offer a two- or four-year degree. There is intense competition among signal providers, who compete not only on the fees they charge but also on the choice of monitoring technologies---the mapping of potential employees' efforts into job market signals. For instance, one source estimates that over 1 million distinct credentials are currently available in the USA alone.\footnote{See \url{https://credentialengine.org/all-resources/counting-credentials/}.} 
These choices shape the signaling options available to workers and influence wages in the job market. Further, their existence raises important economic questions. For instance, would signal providers fuel an unproductive arms race among workers in pursuing their own profit? Is competition among them socially beneficial?

We answer these questions within a version of the Spence model that endogenizes two critical features: signaling technology and access cost. Unlike the traditional approach, our framework recognizes that signal providers \emph{actively choose} both their attendance fees and monitoring technologies.
 As in \cite{spence1973job}, we have candidates who are privately informed about their ability and a competitive job market. We study both the case of a monopolist certifier (i.e., a school) and competition among $n$ ex-ante identical schools that determine \emph{both} attendance fees and signaling technology. We explore what signaling opportunities will be available to job candidates and what outcome(s) will emerge when both costs and monitoring technologies are endogenously determined by profit-maximizing school(s).

We show that, in the unique equilibrium in our setting, a profit-seeking monopolist school offers an {\em uninformative} signaling technology, generating the same signal regardless of effort,\footnote{An example of which would be grade non-disclosure policies at several Business schools.} and capture all the surplus in terms of fees. In this case, the separating equilibrium envisaged in \cite{spence1973job} vanishes, and no wasteful effort is exerted! Competition among schools affects both the total surplus produced and its distribution. On the one hand, competition in fees allows students to retain a higher share of the total surplus. On the other hand, competition to attract high types induces schools to offer more informative monitoring policies, incentivizing higher levels of wasteful effort. 
Absent frictions, competition harms social welfare. However, in the presence of credit constraints, competition can generate higher total surplus than a monopoly.  

Our game proceeds as follows. First, the schools (or school) publicly and simultaneously commit to a signaling policy comprising an attendance fee and a monitoring policy, i.e., a mapping from students' efforts to market signals. We focus on the case of a deterministic monitoring policy, i.e., where the signal sent by the school is a deterministic function of the effort taken by the student.\footnote{Extending to the case of a stochastic monitoring policy would not alter the main conclusions of the paper (see Section \ref{Discussion}) but would require a significant extension of our main refinement. A more detailed discussion is available upon request.
} 
Students, privately informed about their productivity, then choose which school, if any, to attend. Upon selecting a school, they pay the fee and decide how much effort to exert. Based on this effort, a signal is generated according to the chosen school's monitoring policy. The job market observes, for each candidate, which school they attended (if any) and the resultant signal. As in the standard Spence model, the market offers every student a wage equal to the student's expected productivity, conditional on the generated signal. Our main deviation from the Spence model is that we explicitly model profit-maximizing schools that can commit to a signaling policy. The subsequent game proceeds akin to Spence.

As standard in signaling models, our game displays a multiplicity of perfect Bayesian/ sequential equilibria due to off-path beliefs. In our setting, this multiplicity is exacerbated by the schools' endogenous choices. Even with a monopolist school, different signaling policies can arise in equilibrium, sustained by different equilibria in the continuation (signaling) subgames. To make progress, we invoke an equilibrium refinement that extends \cite{cho1987signaling}'s D1 criterion to our setting, which allows us to discipline off-path play sufficiently to obtain sharp equilibrium predictions. 

Our analysis aims to evaluate how competition, credit constraints, and possible government policies affect the total social welfare and its distribution. We capture two distinct sources of inefficiencies that can reduce social welfare away from its maximum: wasteful effort and sub-optimal job assignment.  
In particular, we consider two distinct signaling purposes: \emph{sorting}, where all workers are valuable for the firms, and \emph{screening}, where low-ability workers are detrimental (inflict losses even at zero wages). In the sorting case, signaling serves no social purpose, and thus, any credentialing system that induces no wasteful effort is socially optimal. In the screening case, however,  social optimality also requires firms to hire only high-ability students. 

In the monopoly scenario, a unique equilibrium outcome satisfies our refinement. In it, the school maximizes and fully appropriates social surplus: no wasteful effort is exerted, all and only productive types are employed, and the fees capture the entire surplus.
In the sorting context, where all workers are productive, the school adopts an uninformative monitoring policy (pooling all effort levels) and charges a sufficiently low fee to ensure full enrollment  (Proposition \ref{prop:mon_sorting}). 
In contrast, in the screening context, where employing low types is socially detrimental, the school can ensure all high types enroll at any given fee by enabling them to signal effort levels that only they are willing to sustain. As the fee increases, these effort levels decrease, eventually converging to zero. Ultimately, we show that the school still maximizes and fully extracts social welfare by ensuring only high types enroll and no wasteful effort is exerted (Proposition \ref{prop:mon_screen}).

We then analyze the case where $n$ identical schools compete. Here, competing schools can employ two strategies to lure students away from each other: lowering fees and adjusting monitoring policies. Lower fees attract all types equally but do not affect social welfare; more informative monitoring policies appeal to high-type students seeking to differentiate themselves from low-type students through increased effort. Consequently, even when fees are driven to zero, competition for high-type students compels schools to offer more informative monitoring, prompting high-types to exert more effort to signal their ability.  This dynamic leads to wasteful effort and lower social welfare.

In general, even with our refinement, a plethora of equilibria can arise in the setting with competition. However, one particular equilibrium outcome stands out for its robustness: full separation, yielding the Riley outcome (Propositions \ref{prop:R} and \ref{prop:competition_fierce})---low-type students exert no effort and high-type students exert the lowest effort that low-types are unwilling to emulate. Thus, compared to the monopoly case, competition shifts surplus from the school to high-type students and entails efficiency losses due to socially wasteful effort in equilibrium. The critical driver of this result is the high-type students' willingness to pay to access a monitoring policy that enables them to better distinguish themselves from low-type students and thus obtain higher wages.

Our analysis yields the following takeaways. Firstly, the choice to model signaling opportunities available to students as exogenously given is with loss. In contrast with the predictions of the literature on signaling, we demonstrate that the efficient outcome is attainable. In particular, a profit-seeking monopolist school optimally induces a social optimum. However, all surplus accrues to the monopolist signaling school. Competition improves the distribution of surplus at the cost of being inefficient. In particular, the signaling policy chosen by competing schools in equilibrium necessarily results in inefficient (non-zero) effort expended by the student. Therefore, the usual economic intuition that competition enhances economic efficiency fails in this context.  
The role of competition is partially restored when market frictions, such as credit constraints, lead the monopolist school to expand enrollment to maximize profits, pooling unproductive and productive types. In such cases, competition may generate a higher total surplus than monopoly due to a more efficient job allocation. 

\subsection{Related Literature}

As noted earlier, we depart from the literature on signaling inspired by \cite{spence1973job} by {\em endogenizing} the mapping from student effort to signal observed by the market. We assume that profit-maximizing intermediaries design and offer these mappings for a fee. Our objective is to understand which mappings emerge in equilibrium depending on the level of competition in the intermediaries' market.
Our focus is not on the design of this mapping to achieve various ends (maximize effort, maximize a student's probability of being employed)  as in, say,   \cite{ostrovsky2010information}, \cite{popov2013university}, \cite{boleslavsky2015grading}, and
\citet{olszewski2019pareto}. 

Our schools are analogous to the information designers in \cite{kamenica2011bayesian} or \cite{bergemann2016information}. They differ in that they do not observe the relevant state (the student's type), but an endogenous quantity (effort taken by the student which depends on the signaling policy the school commits to). In this sense, our setting is similar to the literature on information design with moral hazard, e.g., \cite{boleslavsky2020bayesian} or \cite{georgiadis2020optimal}. However, our schools are not interested in incentivizing effort by the student per se. Their goal is to influence the employment market's beliefs about the student's underlying type, which influences the wage they offer and, as a result, the fees the student is willing to pay.

The closest paper to ours is probably \cite{lu2019selling}, which considers a profit-maximizing school selling signals to a mass of heterogeneous students.\footnote{\citet{rayo2013monopolistic} also examines the problem of a monopolist intermediary selling signals to informed agents whose payoff is their conditional expected types. } The school 
does not choose a monitoring policy but a menu specifying \emph{effort-contingent} fees, assuming all students are productive (analogous to our sorting case). His focus is on the welfare impact of price transparency; ours is on the effect of competition on welfare and monitoring technologies.

Schools in our model are similar to but not identical to certifiers in the literature on optimal certification \citep{lizzeri1999information,albano2001strategic, 
zubrickas2015optimal,
harbaugh2018coarse, 
demarzo2019test, boleslavsky2020bayesian, bizzotto2021optimal, asseyer2024certification}. Schools, unlike certifiers, cannot directly observe the students' types. They only observe wasteful effort, whose cost depends on the students' types.  This assumption dramatically changes the school's incentives and social welfare considerations.

\section{Model}

There are three types of agents: a unit mass of students (senders) $S$, a mass of firms (receivers) $R$, and a set of $n\geq 1$ schools (intermediaries) $I=\{1,...,n\}$. 

\medskip
\noindent \textbf{(Standard) Notation:} For a function $g: A \times B \to C$, let $g_a: B \to C$ denote the projection of $g$ with respect to $A$, where $g_a(b) = g(a, b)$ for each $b \in B$ and fixed $a \in A$. Furthermore, for any subset $J \subseteq A \times B$, denote by $g(J)$ the image of $J$ under $g$, i.e., $g(J) := \{ c \in C \mid g(a, b) = c \text{ for some } (a, b) \in J \}$. 

\medskip
\noindent \textbf{Students:}
Students privately know their productivity, which can be either low or high: $\theta\in \Theta = \{\theta_L, \theta_H\}$, with $\theta_L < \theta_H$ and $\theta_H > 0$. We denote by $\lambda \in (0,1)$ the probability (or population fraction) of high-type students. 
As in standard signaling models, students can exert effort at a cost that depends on their productivity type. A type $\theta$ student who chooses $e\in \Re_{+}$  incurs  a cost $c:\Theta\times$ $\Re_{+}\rightarrow \Re_{+}.$ The function $c$ is strictly increasing and continuous in the second argument, decreasing in the first argument, and satisfies strict decreasing differences, i.e.,  
\begin{align*}
\forall e'> e: 
0 \leq c\left(  \theta_{L},e\right)  -c\left(  \theta_{H},e\right)  < c\left(
\theta_{L},e^{\prime}\right)  -c\left(  \theta_{H},e'\right) .
\end{align*}
We normalize $c\left(
\theta,0\right) =0.$ Unlike standard signaling models, however, students' efforts are not directly observed. The school determines the mapping from effort to signal.

\medskip 
\noindent \textbf{Schools:}
We model schools as profit-maximizing intermediaries. Each school $i \in I$ commits to a  policy $p_i=(f_i, M_i)\in \Re_{+} \times \mathcal{M} $, consisting of an attendance fee $f_i \in \R_+$ and a right-continuous monitoring policy $M_{i}:\Re_{+}\to \mathbb{M}$, mapping every effort $e \in \R_+$ into a message $ m =M_i\left( e \right)\in\mathbb{M}$. 
Without loss of generality, we assume all schools can rely on the same (rich) signal space $\mathbb{M}$.\footnote{We make the expositional choice to assume a predefined rich set of possible messages $\mathbb{M}$, ensuring that the set of feasible strategies for every school is well-defined. The alternative, i.e., letting the school select a set of possible messages as part of its strategies, would result in the strategy space being a set of all feasible sets, an ill-defined object.
}  
Denote by $\mathbf{\mathbf{p}}= (p_i)_{i \in I} \in 
(\R_+ \times \mathcal{M})^n \equiv \mathcal{A}$ the tuple of strategies chosen by the schools. 
\medskip

\noindent \textbf{Firms:} Firms operate in a perfectly competitive market, offering wages to attract students.
Type $\theta$ students produce a value $\theta$ for the firm that hires them. In the sequel, we discuss both the case where $\theta_L < 0 < \theta_H$ (\emph{screening}), i.e., the firm only wants to hire high types, and the case where $0 < \theta_L < \theta_H$ (\emph{sorting}), i.e., the firm wants both types but at different wages. For every student, firms observe only the school attended, $i \in I$, and the message generated, $m \in M_i(\Re_+)$. Firms do not hire students who have not attended a school.\footnote{One interpretation is that a diploma is necessary to access the job. Alternatively, we can think that students are unproductive until they have attended a school. In this second interpretation, school attendance allows students to attain their productivity type $\theta$, while the effort $e$ (e.g., exam preparation) remains wasteful.}

\medskip
\noindent \textbf{Timing:}
The timing is the following:
\begin{enumerate}
\item Simultaneously, each school $i\in I$ publicly commits to a policy $p_i = (f_i, M_i)$. 

\item Students privately observe their types, $\theta \in \Theta$, and, given the policy vector $\mathbf{\mathbf{p}}$, choose whether to attend one of the schools or opt for their outside option of $0$. 
Enrolled students also choose how much effort $e\in \Re_+$ to exert.
\item Each school $i$ collects a fee $f_i$ from every enrolled student, observes their effort level $e \in \mathbb{R}_+$, and generates the corresponding message, $m = M_i(e)$.
\item Firms observe the school attended (if any) and the message $m$ generated by each student, then simultaneously make wage offers.
\item Each student accepts the highest wage offer, and payoffs are realized.
\end{enumerate}

\medskip
\noindent \textbf{Payoffs:} The payoff of a type $\theta$ student from attending school $i$, paying fee $f_i$, exerting effort $e$, and obtaining wage $w$ is
\[
w-f_i  -c \left(  \theta, e\right)  .
\]
The outside option of not attending any school is assumed to be $0$.
The profit of school $i$, charging fee $f_i$ and attracting a mass $\beta_i\in [0,1]$ of student, is $\pi_i=\beta_i f_i $. Finally, a firm hiring a worker of type $\theta$ at wage $w$ earns a net payoff of $\theta-w$. 

\medskip
\noindent \textbf{Strategies:}
A pure strategy for school $i$ consists of a policy, $p_i=(f_i,M_i) \in \R_+ \times \mathcal{M}$.

A student’s strategy is a function $\psi:  \mathcal{A} \times \Theta \to \Delta((I \cup \{0\}) \times \mathbb{R}_+)$ that determines school and effort choices based on type and policies. Given policies $\mathbf{p}$, $\psi_{\mathbf{p}, \theta}(i, e)$ is the joint probability that a type-$\theta$ student attends school $i \in I$ (or no school if $i=0$) and exerts effort $e \in \mathbb{R}_+$. The marginal probability of enrollment in school $i$ for type $\theta$ is given by $\psi_{\mathbf{p}, \theta}(i) = \int_{e \in \mathbb{R}_+} \psi_{\mathbf{p}, \theta}(i, e) \, \mathrm{d}e$. 
The marginal probability of effort $e$ for type $\theta$ is $\psi_{\mathbf{p}, \theta}(e) = \sum_{i\in (I \cup \{0\})} \psi_{\mathbf{p}, \theta}(i, e)$. Finally, when $\psi_{\mathbf{p}, \theta}(i)>0$ we denote the probability that type $\theta$ student exerts effort $e$ conditional on enrolling in $i$ by $\psi_{\mathbf{p}, \theta}(e|i) = {\psi_{\mathbf{p}, \theta}(i, e)}/{\psi_{\mathbf{p}, \theta}(i)} $.

A strategy for a firm consists of a wage schedule $\omega: \mathcal{A} \times I \times \mathbb{M} \rightarrow \R_+ \cup \{\emptyset\}$, where  $\omega(\mathbf{\mathbf{p}}, i, m) \in \R_+$ denotes the wage offered to a student who attended school $i\in I$ and generated message $m\in \mathbb{M}$ under the policy vector $\mathbf{p}$, and $\omega(\mathbf{\mathbf{p}}, i, m)=\emptyset$ indicates that no offer is made to that student. Finally, let $\mu(\theta|\mathbf{p}, i, m) \in [0,1]$, with $\mu: \mathcal{A} \times I \times \mathbb{M} \rightarrow \Delta \Theta$, denote the firms' belief that a student is of type $\theta$, given attendance at school $i$ and message $m$ under the policy vector $\mathbf{p}$.

Given $(\mathbf{p}, \psi, \omega)$, a type $\theta$ student's expected payoff is given by
\[
U_\theta (\mathbf{p}, \psi, \omega) = \sum_{i \in I} \int_{e \in \mathbb{R}_+} \left[ \omega({\mathbf{p}},i, M_i(e)) - c(\theta, e) -f_i \right] \, \psi_{\mathbf{p}, \theta}(i, e) \, \mathrm{d}e, 
\]
and school $i$'s profit is 
\[
\pi_i (\mathbf{p}, \psi, \omega) = f_i\, \left[\lambda \psi_{\mathbf{p}, \theta_H}(i)+(1-\lambda)\psi_{\mathbf{p}, \theta_L}(i)\right].
\]

\subsection{Equilibrium}
The set of Perfect Bayesian Equilibria (PBE), $\mathcal{E}^*=(\mathbf{p}^*, \psi^*, \omega^*, \mu^*)$, of this game is immense. As standard in signaling models, every subgame after the schools' choices admits a multiplicity of Perfect Bayesian/ sequential equilibria. This multiplicity propagates upstream as different equilibria in off-path signaling subgames can support equilibria involving different policies. 
For example, in the screening case ($\theta_L<0$), every school is willing to offer any monitoring policy in equilibrium, believing that failing to do so will result in zero enrollment, with students being recognized as low types if they choose to enroll.

In the subgame following schools' policy choices, we focus on equilibria where off-path beliefs satisfy an extended D1 refinement in the spirit of \cite{cho1987signaling} and \cite{banks1987equilibrium}. Like other refinements based on forward induction, it restricts the possible beliefs of the firms off the equilibrium path (i.e., after an observable deviation by schools or students). Intuitively, the D1-criterion prescribes that receivers (firms) should attribute any off-path message $m$ to the sender's (student's) types being the $\theta$ \textit{most likely} to deviate and send that message. That is, to those $\theta$s for whom the set of receivers' responses to $m$ that makes them prefer generating $m$ over adhering to their candidate equilibrium strategy (and obtaining their candidate equilibrium payoff) is a \textit{strict superset} of the set of responses that makes any other type better off.

There are two main challenges in directly applying the D1 criterion to our sub-games. Unlike the classical setting:
\begin{enumerate}
\item Our model allows for general monitoring policies where effort may not be directly observable; the mapping from efforts to signal is deterministic but can be arbitrarily coarse.
\item Senders can choose among different monitoring policies at different fees. 
\end{enumerate}
Therefore, we need to extend the classical D1-refinement.
 
In the subgame following schools' policies $\mathbf{p} \in \mathcal{A}$, the PBE  $ \mathcal{E}^*_{\mathbf{p}}=(\psi^*_{\mathbf{p}}, \omega^*_{\mathbf{p}}, \mu^*_{\mathbf{p}})$ consists of three elements: (i) the student strategy $\psi^*_{\mathbf{p}}$, mapping student types into their school and effort choices; (ii) the firms' belief response $\mu^*_{\mathbf{p}}$, assigning the probability $\mu^*_{\mathbf{p}}(\theta | i, m)$ that a student attending school $i$ and generating signal $m$ is of type $\theta$; and (iii) the firms' wage response $\omega^*_{\mathbf{p}}$, assigning the wage $\omega^*_{\mathbf{p}}(i, m)$ to students attending school $i$ and generating signal $m$. 

The definition of PBE in this subgame is straightforward and omitted: all players' strategies are mutual best responses, firms' wage offers upon observing a message (on- or off-path) are consistent with their beliefs, 
\[\omega^*_{\mathbf{p}}( i, m)=\theta_H \, \mu^*_{\mathbf{p}}(\theta_H|i,m) +\theta_L[1-\mu^*_{\mathbf{p}}(\theta_H|i,m)],\]
whenever the right-hand side is non-negative, and $\emptyset$ otherwise;
and firms' on-path beliefs follow Bayes' rule. 

Given the $\mathcal{E}^*_{\mathbf{p}}$-equilibrium strategy profile of students $\psi^*_{\mathbf{p}}$, let $S^{\psi^*_{\mathbf{p}}} \subseteq I \times \mathbb{M}$ denote the set of student signals (i.e., chosen school and generated message) with positive probability: 
$S^{\psi^*_{\mathbf{p}}}=\left\{(i,m) \in I \times \mathbb{M}:\psi^*_{\mathbf{p}}(i,e)>0, \, M_i(e) = m \right\}.$ 
Also, denote by $U^{\mathcal{E}^*_{\mathbf{p}}}(\theta)$ the expected utility of type $\theta$ from playing their equilibrium strategy  $\psi^*_{\mathbf{p}}$ when schools announce $\mathbf{p}$ and firms respond with  $\omega^*_{\mathbf{p}}$. 
Given $\mathbf{p}$, firms can detect students' deviations from $\psi^*_{\mathbf{p}}$ only if they observe a message $(i,m) \notin S^{\psi^*_{\mathbf{p}}}$. Our extended D1 refinement on firms' beliefs applies precisely to off-path messages in every subgame.

In particular, given a candidate equilibrium $ \mathcal{E}^*_{\mathbf{p}}=(\psi^*_{\mathbf{p}}, \omega^*_{\mathbf{p}}, \mu^*_{\mathbf{p}})$ in the (possibly off-path) subgame following $\mathbf{p}$, define 
\begin{align*}
&D^\geq_{\mathcal{E}^*_{\mathbf{p}}}(\theta, i,m) = \Big\{w \in [\max\{0,\theta_L\},\theta_H]:U^{\mathcal{E}^*_{\mathbf{p}}}(\theta)\leq w-f_i -\min_{e:M_i(e)=m} 
c(\theta,e) 
\Big\},\\
&D^>_{\mathcal{E}^*_{\mathbf{p}}}(\theta, i,m) =  \Big\{w \in [\max\{0,\theta_L\},\theta_H]:U^{\mathcal{E}^*_{\mathbf{p}}}(\theta)< w-f_i -\min_{e:M_i(e)=m} 
c(\theta,e) 
\Big\}.
\end{align*}
In words, $D^\geq_{\mathcal{E}^*_{\mathbf{p}}}(\theta, i, m)$ denotes the set of wages $w$ for which type $\theta$ students would weakly prefer deviating to exert an effort $e \in M_i^{-1}(m)$ in school $i$, rather than following the equilibrium strategy $\psi^*_{\mathbf{p}}$ and receiving the expected payoff $U^{\mathcal{E}^*_{\mathbf{p}}}(\theta)$.\footnote{We restrict attention to sequentially rational wages, i.e., a wage offer that would be made for some belief a firm may have about the student's type.} Meanwhile, $D^>_{\mathcal{E}^*_{\mathbf{p}}}(\theta)$ is the set of wages that make these deviations strictly beneficial.
\begin{definition}\label{df:refine}
A Perfect Bayesian Equilibrium in the subgame following $\mathbf{p}$, $\mathcal{E}^*_{\mathbf{p}}=(\psi^*_{\mathbf{p}}, \omega^*_{\mathbf{p}}, \mu^*_{\mathbf{p}})$ fails our extended D1 criterion, if there exists an off-path message $(i,m)\notin S^{\psi^*_{\mathbf{p}}}$, and types $\theta, \theta'\in \Theta$, 
such that, $\mu^*_{\mathbf{p}}(\theta|i,m)>0$, and $$D^\geq_{\mathcal{E}^*_{\mathbf{p}}}(\theta, i,m)\subsetneq D^>_{\mathcal{E}^*_{\mathbf{p}}}(\theta', i,m).$$
 We call $EPBE$ any $PBE$ of the subgame that satisfies our extended D1 criterion.
\end{definition}
Our refinement prescribes that, upon observing a student unexpectedly enrolling in school $i$ and generating message $m$, firms should exclude the possibility that the student is of type $\theta$ if the set of wages justifying type $\theta$'s deviation to choices consistent with $(i, m) \notin S^{\psi^*_{\mathbf{p}}}$ is strictly smaller (in the sense of set inclusion) than that for type $\theta'$.

\bigskip
For example, consider the subgame following $\mathbf{p}=(f_j,M_j)_{j\in I}$, where school $i$ offered a (cutoff) monitoring policy $M_i$: 
\[
M_i(e) = \begin{cases}
   m \quad &\text{if } e \geq  \hat{e},\\
   0 \quad &\text{otherwise,} 
\end{cases}
\]
with $\hat{e} > 0$ and $m \neq 0$. 
Suppose the candidate equilibrium payoff 
$\mathcal{E}^*_{\mathbf{p}}=(\psi^*_{\mathbf{p}}, \omega^*_{\mathbf{p}}, \mu^*_{\mathbf{p}})$ specifies that no student enrolls in school $i$ and exerts effort $ e \geq  \hat{e}$, implying $(i,m) \notin S^{\psi^*_{\mathbf{p}}}$.
Our refinement requires the receiver's belief following the off-path signal $(i,m)$ to assign zero weight to type $\theta$ if there exists $\theta' \in \Theta \setminus {\theta}$ such that $D^\geq_{\mathcal{E}^*_{\mathbf{p}}}(\theta, i,m)\subsetneq D^>_{\mathcal{E}^*_{\mathbf{p}}}(\theta', i,m)$.
In this case, this is equivalent to:
\begin{align}\label{ineq:D1}
     \min\left\{\theta_H-f, U^{\mathcal{E}^*_{\mathbf{p}}}(\theta) + c(\theta, \hat{e})\right\} > U^{\mathcal{E}^*_{\mathbf{p}}}(\theta') + c(\theta', \hat{e}).
\end{align}
We emphasize that this is the primary implication of the refinement used in our proofs.

Suppose there is only one school ($n = 1$), and this school charges no fees while adopting a perfectly informative monitoring policy ($M(e) = e$ for all $e \in \mathbb{R}_+$). In that case, our extended D1 refinement coincides with the classical D1 refinement in \citet{banks1987equilibrium}. A natural concern is whether such refinement is too strong, potentially ruling out all PBEs in the subgame following some $\mathbf{p}$. The following theorem asserts that an EPBE exists in every subgame.\footnote{Proofs are omitted in the main text and can be found in the appendix.}

\begin{theorem}[Existence of EPBE]\label{thm:existence}
An EPBE (i.e., a PBE that satisfies our extended D1 criterion) exists in the subgame following any policy profile.
\end{theorem}

We are now prepared to refine the set of perfect Bayesian equilibria for the entire game. Even when we require that every equilibrium in the subgame is an EPBE, a trivial source of multiplicity in equilibrium outcomes persists due to schools setting up strategically irrelevant messages—such as those sent only in response to implausibly high effort levels (resulting in negative payoffs for all student types even when rewarded with $w=\theta_H$). To address this, we introduce a minimality requirement: among the PBEs that survive the above refinement and produce the same outcome in terms of students' actions, messages sent, firms' wages, and school fees, we define \textit{Refined Perfect Bayesian Equilibria (RPBE)} as those involving the minimum number of unsent messages. 
\begin{definition}
    A PBE of our game $ \mathcal{E}^*=(\mathbf{p}^*, \psi^*, \omega^*, \mu^*)$ is a Refined PBE (RPBE) if:
    \begin{enumerate}[(i)]\itemsep0pt 
        \item \textbf{EPBE in every subgame:} for every $\mathbf{p}\in \mathcal{A}$,  $\mathcal{E}^*_{\mathbf{p}}=(\psi^*_{\mathbf{p}}, \omega^*_{\mathbf{p}}, \mu^*_{\mathbf{p}})$ is an EPBE of the subgame following $\mathbf{p}=(f_i,M_i)_{i\in I}$.
        \item \textbf{Minimality:} there is no PBE $\mathcal{E} = (\mathbf{p}, \psi, \omega, \mu)$ satisfying (i) above such that $\psi_\mathbf{p} = \psi_{\mathbf{p}^*}$, $(f_i)_{i \in I} = (f_i^*)_{i \in I}$, and, for all $(i, e) \in I \times \mathbb{R}_+$ such that $\psi_{\mathbf{p}^*}(i, e) > 0$, $M_i(e) = M_i^*(e)$ and $\omega(\mathbf{p}, i, M_i(e)) = \omega^*(\mathbf{p}^*, i, M^*_i(e))$; \textbf{but} $M_i(\mathbb{R}_+) \subseteq M_i^*(\mathbb{R}_+)$ for all $i \in I$, with $M_i(\mathbb{R}_+) \subsetneq M_i^*(\mathbb{R}_+)$ for some $i \in I$.
    \end{enumerate}
\end{definition}

\medskip
\subsection{Welfare} \label{sec:welfare_def}
The (utilitarian) social welfare of an RPBE $\mathcal{E}^*$ is simply 
$$ \lambda U^{\mathcal{E}^*_{\mathbf{p}^*}}(\theta_H) +  (1-\lambda) U^{\mathcal{E}^*_{\mathbf{p}^*}}(\theta_L) + \sum_{i\in I} \pi^*_i. $$
It can be rewritten as the total productivity from the job allocation net of the (wasteful) cost of effort, e.g.,
\[
\begin{split}
    \lambda \theta_H \mathrm{Pr}(\text{employed}|\theta_H) + (1-\lambda) \theta_L \mathrm{Pr}(\text{employed}|\theta_L) \\
    - \left[\lambda \int_{\mathbb{R}_+} c(\theta_H, e) \psi^*_{\mathbf{p}^*, \theta_H}(e) \mathrm{d}e + (1-\lambda) \int_{\mathbb{R}_+} c(\theta_L, e) \psi^*_{\mathbf{p}^*, \theta_L}(e) \mathrm{d}e\right].
\end{split}
\]
In the sorting case, hiring students of any type always enhances the total productivity from the job allocation. However, in the screening case  ($\theta_L <0$), 
the negative term of $\theta_L \mathrm{Pr}(\text{employed}|\theta_L)$ reflects a source of welfare loss arising from hiring low-type students.

In what follows, we refer to \emph{maximum social welfare} as the largest possible social welfare achievable given the parameters. Informally, in the sorting case ($\theta_L > 0)$, this requires that all students go to school and are employed; in the screening case ($\theta_L<0)$, this requires that only high types go to school and are employed. In both cases, no effort should be taken since education is unproductive. 

\section{Monopoly} 
In this section, we examine the case of a monopolist school ($n=1$). This is the natural counterpart to the standard costly signaling model \citep{spence1973job}, with one key difference: in our model, the monopolist school designs and sells the signaling technology to maximize its profits. We explore two distinct settings.
\begin{enumerate}[(i)]\itemsep0pt
\item \textbf{Sorting,} where both types are productive ($\theta_H > \theta_L > 0$). In this case, signaling has no social value and only functions to align students' wages with their productivity. 
\item \textbf{Screening,} where low-type students have negative productivity when hired ($\theta_H > 0 > \theta_L$). In this case, signaling has potential social value, as the maximum is achievable only when low-type students remain unemployed.
\end{enumerate}

\subsection{Sorting}\label{sec:mon_sorting}

In the sorting case, we show that a monopolist school, $i=1$, can extract the full potential surplus, $\E \theta$ while producing no information about students' efforts. Indeed, by charging an attendance fee of $f = \E \theta$ in exchange for an uninformative monitoring policy,
the school manages to attract all students.

\begin{proposition}[Monopoly: Sorting] \label{prop:mon_sorting}
Suppose $\theta_H  > \theta_L\geq 0$, and there is a monopolist school. The RPBE outcome is unique and features full pooling. $ \mathcal{E}^*=(p^*, \psi^*, \omega^*, \mu^*)$ is such that:
\begin{enumerate}[(i)]\itemsep0pt
    \item The school offers an uninformative monitoring policy $M^*$: $M^*(e) = m\in \mathbb{M}$ for all $e\in \R_+$.
    \item The school charges $f^*= \E \theta$.
    \item All students enroll and exert no effort on path: 
    $\psi^*_{p^*, \theta}(i =1,e=0)=1$  
     for all $\theta\in\{\theta_L,\theta_H\}$.
     \item Firms pay all graduates a wage $\omega^*_{p^*}(i = 1,m)=\E \theta$ on path.
\end{enumerate}

\end{proposition}

We outline the proof below, which is particularly instructive, as its basic structure is used for subsequent results. At a high level, the proof consists of three main steps: (1) showing that any equilibrium outcome that extracts all potential surplus satisfies $(i)-(iv)$, (2) verifying that this equilibrium outcome can indeed be supported as an RPBE, and (3) demonstrating that no other equilibrium outcome survives the refinement. The proof also underscores the role of our refinement in ensuring tractability by yielding a unique equilibrium.

\begin{proof}[Proof of Proposition \ref{prop:mon_sorting}]
Denote by $ \mathcal{E}^*=(p^*, \psi^*, \omega^*, \mu^*)$ an RPBE of the monopolist game, by $\mathcal{E}^*_p$ the induced EPBE in the subgame following $p$, and by $\pi^*_p$ the associated school's profits in that subgame.

\paragraph{Step 1}: \emph{$\pi^*_p \leq \E \theta$ for all $p\in \mathcal{A}$, and $\pi^*_{p^*} = \E \theta$ implies $\mathcal{E}^*_{p^*}$ must satisfy (i)-(iv).}
The amount extracted in fees in $\mathcal{E}^*_p$ cannot exceed the surplus generated in the downstream market, which is at most $\E \theta$. Thus $\pi^*_p \leq \E \theta$. Moreover, the surplus generated in the downstream market is $\E \theta$ only if all students enroll and exert zero effort; thus (ii)-(iv) need to hold if $\pi^*_{p^*} = \E \theta$. 
Finally, the RPBE minimality requirement implies that the monitoring policy must be uninformative in equilibria where no effort is exerted; thus, (i) must also hold.

\paragraph{Step 2}: \emph{There exists an RPBE such that (i)-(iv) hold.}
Consider a PBE $ \mathcal{E}^*=(p^*, \psi^*, \omega^*, \mu^*)$ such that $(i)-(iv)$ holds on-path (in $\mathcal{E}^*_{p^*}$) and such that the equilibria $\mathcal{E}^*_{p'}$ prescribed in the subgames following any $p'\neq p^*$ are EPBE.\footnote{By Theorem \ref{thm:existence}, there exists an EPBE in the subgame following any policy $p\in \mathcal{A}$.} 
We need to show that $\mathcal{E}^*$ is an RPBE.
First note that $\pi^*_{p*}=\E \theta\geq \pi^*_{p'}$ (since $\mathcal{E}^*_{p^*}$ satisfies $(i)-(iv)$), and thus the school has no profitable deviation. Further, since students can generate only one signal under $p^*$ ($p^*$ is uninformative (i)), no message is unsent in the subgame following $p^*$. Thus, also $\mathcal{E}^*_{p^*}$ is a EPBE (it vacuously satisfies the extended D1 refinement).  Finally, since $\mathcal{E}^*$ also satisfies the minimality requirement, we can conclude that $\mathcal{E}^*$ is an RPBE.

\paragraph{Step 3}: \emph{In every RPBE $\pi^*_{p^*} = \E \theta$.} 
Since (by Step 1) $\pi^*_{p^*}\leq \E \theta$, we just need to show that it cannot be that $\pi^*_{p^*}<\E \theta$. 
For the sake of contradiction,  suppose there exists an RPBE $\mathcal{E}'=(p', \psi', \omega', \mu')$ where the school earns $\pi'_{p'} < \E \theta$. 

Consider a deviation by the school from policy $p'$ to policy $\hat{p} = (\hat{f}, \hat{M})$ such that:
\begin{align*}
    \hat f &= \E \theta - c(\theta_L, \epsilon),\\
    \hat M (e) &=
    \begin{cases}
         \epsilon \quad &\text{if } e \geq \epsilon,\\
         0 \quad &\text{if } e < \epsilon,
    \end{cases}
\end{align*}
where $\epsilon > 0$ is sufficiently small, i.e., such that $c(\theta_L, \epsilon) <\min\left\{ \theta_H-\E\theta, \E \theta - \pi'_{p'}\right\}$.

We now show that in the subgame following $\hat{p}$, the EPBE $\mathcal{E}'_{\hat{p}} = (\psi'_{\hat{p}}, \omega'_{\hat{p}}, \mu'_{\hat{p}})$, must involve all students enrolling. For $\epsilon$ sufficiently small, this generates a profit of $\E \theta - c(\theta_L, \epsilon) > \pi'_{p'}$. Thus, the school has a profitable deviation to $\hat{p}$ from the putative equilibrium policy $p' = (f', M')$, leading to a contradiction.

To understand why, assume for contradiction that in $\mathcal{E}'_{\hat{p}}$, a positive mass of students does not enroll. This would imply $\omega'_{\hat{p}}(1,m=\epsilon) > \E \theta$. Indeed:
\begin{itemize}
    \item[(i)]Suppose $\theta_L$ selects $e = \epsilon$ with positive probability following $\hat{p}$. Then, by the strict submodularity of $c$, $\theta_H$ must enroll and select $e = \epsilon$ with probability 1:  i.e., $\psi'_{\hat{p},\theta_H}(1, e=\epsilon) = 1$. As a result, $\omega'_{\hat p}(1,m=\epsilon ) > \E \theta$.

    \item[(ii)] Suppose instead no students of type $\theta_L$ select $e = \epsilon$. Consider two sub-cases:
    \begin{enumerate}
        \item[(a)] If $\theta_H$-students select $e = \epsilon$ with positive probability, then by Bayes' rule, $\mu'_{\hat{p}}(\theta_H | 1, m=\epsilon) = 1$
        \item[(b)] If no student selects $e = \epsilon$, $\theta_H$ and $\theta_L$ must obtain the same equilibrium payoff under $\mathcal{E}'_{\hat{p}}$, $U^{\mathcal{E}'_{\hat{p}}}(\theta_H) = U^{\mathcal{E}'_{\hat{p}}}(\theta_L)\leq \E\theta$, as they face the same costs for not enrolling or exerting zero effort. 
        Thus, for $\epsilon$ sufficiently small, we have $\max\left\{\theta_H, U^{\mathcal{E}'_{\hat{\mathbf{p}}}}(\theta_L) + c(\theta_L, \epsilon) \right\} > U^{\mathcal{E}'_{\hat{\mathbf{p}}}}(\theta_H) + c(\theta_H, \epsilon)$, and thus our refinement implies $\mu'_{\hat{\mathbf{p}}}(\theta_H \mid m = \epsilon) = 1$.

    \end{enumerate} 
    In either case therefore we have that $\mu'_{\hat{p}}(\theta_H | 1,m=\epsilon) = 1$, and the resulting wage is $\omega'_{\hat{p}}(1,m=\epsilon) > \E \theta$.
\end{itemize}
Finally, observe that under $\hat{\mathbf{p}}$, every student would prefer to enroll, paying $\hat{f} = \E \theta - c(\theta_L, \epsilon)$, exert effort $e = \epsilon$ (incurring a cost of $c(\theta_L, \epsilon)$), and receive a wage $\omega'_{\hat{\mathbf{p}}}(m = \epsilon) > \E \theta$, rather than forgoing enrollment and accepting the zero outside option. Consequently, by offering the policy $\hat{\mathbf{p}}$, the monopolist school ensures full enrollment and secures an equilibrium profit of $\pi'_{\hat{\mathbf{p}}} = \E \theta - c(\theta_L, \epsilon) > \pi'_{p'}$.
\end{proof}

Therefore, under sorting, a monopolist school maximizes and fully appropriates social surplus: no wasteful effort is exerted, all students are employed, and the fees capture the entire social surplus $\E \theta$.

At a technical level, in this case, the sole impact of our refinement is to rule out ``unreasonable'' equilibria in which firms negatively update their beliefs about students' ability upon observing their enrollment in school, regardless of the effort signaled through the monitoring policy. Our refinement ensures that if no students are expected to enroll and exert positive effort following a (possibly out-of-equilibrium) policy $p$, but firms observe that some students do so, they must positively, rather than negatively, update their beliefs about these students' types, due to the cost advantage of high-type students.

\subsection{Screening}
Unlike the sorting case, in the screening case, where $\theta_H > 0 > \theta_L$, low-type students are socially harmful on the job.  Thus, a school's ability to screen students through costly effort is crucial to avoid inefficient outcomes when types are ex-ante indistinguishable by firms. Provided that high-type students have a cost advantage, a sufficiently high fee, combined with the possibility of signaling an infinitesimal amount of effort, suffices to screen the types. As it turns out, a monopolist school can exploit this kind of policy to maximize and extract the entire social surplus. 

\begin{proposition}[Monopoly: Screening]\label{prop:mon_screen} 
Suppose $\theta_H > 0 > \theta_L$, and there is a monopolist school. The RPBE outcome is unique and features full separation: 
\begin{enumerate}[(i)]\itemsep0pt
    \item The school charges $f^*=  \theta_H$ and offers an uninformative monitoring policy $M^*$, such that $M^*(e) = m\in \mathbb{M}$ for all $e\in \R_+$.
    \item High-type students enroll and exert zero effort, $\psi^*_{p^*, \theta_H}(i=1,e=0)=1$; low-type students do not enroll, $\psi^*_{p^*, \theta_L}(i=0,e=0)=1$. 
    \item Firms pay all graduates a wage of $\omega^*_{p^*}(i = 1,m)=\theta_H$.
\end{enumerate}
\end{proposition}
The proof resembles that of Proposition \ref{prop:mon_sorting} and is deferred to Appendix \ref{app:proofs-prop-mon-screen}. 

Therefore, under screening, a monopolist school maximizes and fully appropriates the social surplus: no wasteful effort is exerted, all and only $\theta_H$-students are employed, and the fees capture the entire social surplus $\lambda \theta_H$.

 \section{Competition}

We turn to the case of $n >1$ competing schools. While a monopolist school maximizes social welfare, we show that competition inevitably introduces inefficiencies (wasteful effort) in equilibrium. This is true under both screening and sorting, casting doubt on the social benefits of competition in the market for signaling intermediaries. 

\begin{proposition}[Inefficiency under Competition]\label{prop:competition1}
When $n\geq 2$ schools compete to attract students, no RPBE outcome achieves the maximum social welfare. 
\end{proposition}
Although this result applies to screening and sorting cases, the underlying intuition differs slightly between them. 

In the sorting case, the monopolist achieves maximum social welfare (i.e., zero effort and full employment) by offering an uninformative monitoring policy at a fee of $f^*= \E\theta$. Competing schools, however, must tempt students from each other. Competition causes them to lower their fees and adjust their monitoring policies. 
Lower fees attract all student types equally, redistributing social surplus from schools to students without changing its total. However, even with zero fees, more informative monitoring policies appeal to high-type students seeking to differentiate themselves from low types through increased effort. Consequently, competition for high types drives schools to offer increasingly informative monitoring, prompting high-type students to exert more effort to signal their ability. This dynamic results in wasteful effort and reduced social welfare.

In the screening case, the monopolist achieves the social optimum (i.e., zero effort and full employment of only high-type students) by offering an uninformative monitoring policy and charging $f^* = \theta_H$, discouraging low-type students from enrolling. Competition over fees still shifts welfare from schools to students, but it also incentivizes low types to enroll (and mimic high types), reducing future wages and social welfare due to their negative productivity. Schools adopt more informative monitoring policies to attract high types and deter low types, enabling productive students to signal a greater effort. Although these policies achieve efficient job allocation by deterring low-type enrollment, they lead to wasteful effort and reduced social welfare.

{While all RPBE outcomes are inefficient, they are not equivalent, and one stands out: the RPBE that induces the classical \emph{Riley outcome}---the cheapest separating equilibrium in the standard Spence model. In it, low-type students exert no effort and earn $\max\{0, \theta_L\}$, while high-type students exert the Riley effort $e^R$, which uniquely satisfies $c(\theta_L, e^R) = \theta_H - \max\{\theta_L, 0\}$, and earn $\theta_H$.
}

\begin{proposition}[Riley Outcome]\label{prop:R}Suppose $n \geq 2$ schools compete to attract students. Then, for every $\lambda \in (0,1)$, $\theta_H > 0$, and $\theta_L \in \mathbb{R}$, there exists an RPBE $\mathcal{E}^* = (\mathbf{p}^*, \psi^*, \omega^*, \mu^*)$ such that:

\begin{enumerate}
    \item[(i)] Every school charges $f^* = 0$ and adopts the monitoring policy $M^*$:
    \[
    M^*(e) =  
     \begin{cases}
        m^A & \text{if } e \geq e^R, \\
        m^B & \text{otherwise.}
    \end{cases}
    \]
    
    \item[(ii)] High-type students enroll and exert effort $e^R$: $\psi^*_{\mathbf{p}^*, \theta_H}(i, e^R) = 1/n$ for any $i \in I$.

   \item[(iii)] Low-type students exert no effort:
    \begin{enumerate}[(a)]\itemsep0pt
        \item In the sorting case, they all enroll: $\psi^*_{\mathbf{p}^*, \theta_L}(i, 0) = 1/n$ for every $i \in I$.
        \item In the screening case, none enrolls: $\psi^*_{\mathbf{p}^*, \theta_L}(0, 0) = 1$.
    \end{enumerate}

    \item[(iv)] Firms offer a wage of $\theta_H$ to graduates who generate message $m^A$ and a wage of $\max\{\theta_L, 0\}$ to those who generate message $m^B$: for every $i \in I$,
    \[
    \omega^*_{\mathbf{p}^*}(i, m) =  
     \begin{cases}
        \theta_H & \text{if } m = m^A, \\
        \max\{\theta_L, 0\} & \text{if } m = m^B.
    \end{cases}
    \]
\end{enumerate}
\end{proposition}
Recall that the Riley outcome is the unique equilibrium outcome under standard forward induction refinements (e.g., the intuitive criterion or D1) in the Spence model, where student effort is directly observable to firms. Here, we show that the Riley outcome, which vanishes under monopoly, reemerges as an RPBE when $n \geq 2$ schools compete to attract students, even with endogenous monitoring.

To understand why this separating outcome is particularly robust, note that in any subgame where a school $i$ adopts the prescribed policy $p_i^* = (f^*, M^*)$, there exists an EPBE in which all students exclusively enroll in school $i$, regardless of the policies chosen by other schools.  
First, the school fee $f^* = 0$ cannot be undercut. Furthermore, since high-type students already earn the maximum wage of $\theta_H$ at school $i$, they can only be attracted to another school by the prospect of exerting less effort. However, lower effort appeals more to low-type students than to high-type students. Consequently, our  refinement attributes any off-path message associated with lower effort to low types, resulting in a wage of $\max\{0, \theta_L\}$, which is unattractive to high types. Thus, high types prefer remaining at school $i$, implying that also low types cannot benefit from switching schools. 
As a result, if one school $i$ adopts the policy $(f^*, M^*)$, any deviation by a competing school $j \neq i$ is deterred by the threat that all students would enroll in school $i$, leaving school $j$ with no profit.

While the outcome in Proposition \ref{prop:R} is \emph{``focal''}, other RPBEs emerge depending on the parameters. In what follows, we focus on symmetric RPBEs, i.e., those $\mathcal{E}^* = (\mathbf{p}^*, \psi^*, \omega^*, \mu^*)$ where schools adopt identical policies, and students' enrollment and effort choices are homogeneous across all schools: ${p}^*_i = (f^*, M^*)$ for all $i \in I$, and $\psi^*_{{p}^*, \theta}(i, e) = \psi^*_{{p}^*, \theta}(j, e)$ for any $\theta \in \{\theta_L, \theta_H\}$, $i, j \in I$, and $e \in \mathbb{R}_+$.

In general, the presence of a single competitor is insufficient to drive fees to zero, as Bertrand competition over fees interacts with competition over monitoring policies, resulting in equilibrium multiplicity. For brevity, we focus on settings where competition is sufficiently fierce to ensure zero fees in all RPBEs, deferring the analysis of less intense (``mild'') competition to Appendix \ref{apd:comp}.\footnote{In such cases, alongside the zero-fee equilibrium, a range of equilibria exists where schools charge positive fees.}

\begin{definition}[Fierce Competition]\label{def:fierce} We say that competition among $n\geq 2$ schools is fierce if \textbf{any} of the following conditions holds: (i) $n > {1}/{\lambda}$, (ii) $n \theta_L > \E \theta$, or (iii) $-(n-1)\theta_L \geq \theta_H$.
\end{definition}
While condition (i) applies to both sorting and screening, condition (ii) pertains only to sorting, and condition (iii) exclusively to screening. 

\begin{proposition}[Zero Fees under Fierce Competition]
\label{prop:comp_Riley}
Suppose schools are in fierce competition. In every symmetric RPBE
$\mathcal{E}^* = (\mathbf{p}^*, \psi^*, \omega^*, \mu^*)$, $f^* = 0$.      
\end{proposition}

Intuitively, by undercutting competitors’ positive fees and offering a perfectly informative monitoring policy, a school can attract all high-type students seeking to differentiate themselves from low types through effort. If $\lambda f^* > {f^*}/{n}$ (condition (i)), the profit gained from enrolling all high types outweighs the potential loss from losing some low types. 
In the screening case, this same policy deviation prevents $f^* > 0$ unless a sufficient portion of low-type students enroll. However, enrolling a large fraction of low types is incompatible with positive school profits when low types are sufficiently destructive (condition (iii)). Thus, condition (iii) leads to zero fees. 
Finally, note that the risk of losing low types only applies when $f^* > \theta_L$. In that case, however, the deviating school could attract all students with a fee $\theta_L-\epsilon$. Since schools' total profit is bounded above by the maximal social surplus, which is $\mathbb{E}[\theta]$ in the sorting case, condition (ii) also implies zero fees.

Moreover, while multiple RPBE $\mathcal{E}^*$ with zero fees may coexist, we show (in Proposition \ref{prop:competition_fierce}) that they are all characterized by school policies $p^*$ that support an EPBE yielding the Riley outcome (as in Proposition \ref{prop:R}): low types exert no effort and earn $\max\{0, \theta_L\}$, while high types exert the Riley effort $e^R$, earning $\theta_H$. Indeed, we show that every RPBE requires schools' unilateral deviations to be discouraged by the threat of an EPBE where no student enrolls in the deviating school, low-type students exert no effort and earn $\max\{0,\theta_L\}$, high-type students exert the Riley effort $e^R$ and earn $\theta_H$ (i.e., the Riley outcome in terms of effort and wages).

Thus, every RPBE that does not yield the same outcome as Proposition \ref{prop:R} involves EPBEs where the effort choices of students not enrolling in school $j$ depend directly on school $j$'s policy beyond the "selection" channel (the proportion of high types that do not enroll in $j$). If we rule out this potentially unappealing feature, the outcome in Proposition \ref{prop:R} becomes the unique RPBE outcome of the game. This further supports the idea that the Riley outcome is focal in competitive settings.


\begin{definition}
    [Independence of Irrelevant School] \label{assump_invariant}
An RPBE $\mathcal{E}^* = (\mathbf{p}^*, \psi^*, \omega^*, \mu^*)$ satisfies the IIS property if for any $\mathbf{p}, \mathbf{p}'\in \mathcal{A}$ with $p_i'=p_i$  for every $i\in I\setminus \{j\}$, the EPBEs 
$\mathcal{E}^*_{\mathbf{p}}$ and $\mathcal{E}^*_{\mathbf{p}'}$ are such that: 
if the proportion of high types not enrolling in school $j\in I$ is the same in  $\mathcal{E}^*_{\mathbf{p}}$ and $\mathcal{E}^*_{\mathbf{p}'}$, i.e., if $\psi^*_{\mathbf{p}, \theta}(j)=\psi^*_{\mathbf{p'}, \theta}(j)=1$ for every $\theta\in \Theta$ or 
\[\frac{\lambda \left(1-\psi^*_{\mathbf{p},\theta_H}\left(j\right)\right)}{\lambda \left(1-\psi^*_{\mathbf{p},\theta_H}\left(j\right)\right)+(1-\lambda) \left(1-\psi^*_{\mathbf{p},\theta_L}\left(j\right)\right)}=\frac{\lambda \left(1-\psi^*_{\mathbf{p}',\theta_H}\left(j\right)\right)}{\lambda \left(1-\psi^*_{\mathbf{p}',\theta_H}\left(j\right)\right)+(1-\lambda) \left(1-\psi^*_{\mathbf{p}',\theta_L}\left(j\right)\right)},\]
then the strategy of students that do not enroll in $j$ is the same in  $\mathcal{E}^*_{\mathbf{p}}$ and $\mathcal{E}^*_{\mathbf{p}'}$, i.e., 
\[\psi^*_{\mathbf{p},\theta}\left(i,e|i\neq j\right)=\psi^*_{\mathbf{p}',\theta}\left(i,e|i\neq j\right),\] for any $\theta\in \Theta$, $e \in \R_+$, and $i\in I\setminus \{j\}$.

\end{definition}

In other words, IIS excludes RPBEs where two policy profiles, $\mathbf{p}$ and $\mathbf{p}'$, differ only in school $j$'s policy, yet the resulting EPBEs, $\mathcal{E}_{\mathbf{p}}$ and $\mathcal{E}_{\mathbf{p}'}$, prescribe different choices for students not enrolled in school $j$, even though the proportions of high- and low-type students among them remain identical across the two EPBEs. Intuitively, IIS holds if the EPBE choices of students who do not enroll in school $j$ depend only on their types, the policies offered at schools $i \neq j$, and the proportion of high-type students among those not enrolling in $j$. 
In our proof, IIS solely serves to discipline students' behavior at other schools when one school deviates from the equilibrium policy.

The following proposition characterizes (in both sorting and screening) the set of symmetric RPBE outcomes under fierce competition. It establishes that only the one in Proposition \ref{prop:R}, which mirrors the classical Riley outcome, satisfies the IIS property.

\begin{proposition}[Equilibrium under Fierce Competition]\label{prop:competition_fierce} 
Suppose competition among $n\geq 2$ schools is fierce. Every symmetric RPBE $\mathcal{E}^* = (\mathbf{p}^*, \psi^*, \omega^*, \mu^*)$ falls into one of the following two categories:
\begin{enumerate}
   \item[I.] The RPBE outcome of Proposition \ref{prop:R} (\textbf{Riley outcome}), or

\item[II.] {A Semi-pooling RPBE outcome} where \begin{enumerate}
   \item[(i)] Every school $i\in I$ charges $f^*=0$ and adopts the following monitoring policy
   \begin{align*}
   M^*(e) &=  
   \begin{cases}
        m_A \quad &\text{if } e \geq e^R,\\
        m_B \quad &\text{if } e \in [e_l, e^R),\\
        m_C \quad &\text{otherwise,}
    \end{cases}
\end{align*} 
where $e_l \in [0, e^R)$ and $m_A, m_B, m_C \in \mathbb{M}$;

\item[(ii)] All students enroll; low types exert $e_l$; high types exert $e_l$ with probability $q_h \in (0,1)$ and $e^R$ with probability $1-q_h$: $\psi^*_{p^*, \theta_H}(i,e_l)={q_h}/{n}$, $\psi^*_{p^*, \theta_H}(i,e^R)={(1-q_h)}/{n}$,  $\psi^*_{p^*, \theta_L}(i,e_l)={1}/{n}$ 
     for any $i \in I$;

\item[(iii)] Firms offer the following wage schedule to graduates of any school $i \in I$:
\begin{align*}
   \omega^*_{p^*}(i,m) &= 
   \begin{cases}
        \theta_H \quad &\text{if } m = m_A,\\
        w_l \quad &\text{if } m = m_B,\\
        \max\{\theta_L, 0\} \quad &\text{otherwise,}
    \end{cases}
\end{align*}
 where
\begin{align} \label{eq:mixed wage}
    w_l = \frac{\lambda q_h \theta_H + (1-\lambda) \theta_L}{\lambda q_h  + 1-\lambda} \in (\max\{\theta_L, 0\}, \theta_H).
\end{align}
\end{enumerate}
\end{enumerate}

\noindent Moreover, every symmetric RPBE satisfying IIS induces the RPBE outcome of Proposition \ref{prop:R}.
\end{proposition}

We sketch the key arguments here. First, by Proposition \ref{prop:comp_Riley}, the equilibrium fee is driven to zero under fierce competition. Second, we prove that any separating RPBE must deliver the Riley outcome as in Proposition \ref{prop:R}. Next, we show that the equilibrium outcome cannot be full pooling. Indeed, a unilateral deviation to a small (but positive) fee paired with a perfectly informative monitoring policy ($M_i(e)=e$ for all $e\geq 0$) would profitably attract at least some high-type students who would benefit from distinguishing themselves from low types. 

We then characterize semi-pooling RPBE outcomes with zero fees. We first show that the equilibrium payoff of low-type students must exceed $\max\{\theta_L,0\}$. 
Thus, low types will attend school and exert an effort that does not identify them as low types. Second, we show that the level of effort required to obtain the high wage $\theta_H$ on path must be the Riley effort $e^R$, and thus the school must allow students to signal it. Similarly to the pooling case, the main strategy for proving these results is to analyze a school’s unilateral deviation to a small positive fee and a more informative monitoring policy that allows high types to distinguish themselves from low types without excessive effort.

\section{Welfare Comparison}
How does competition among schools affect welfare in the signaling game? Absent any friction, the monopolist school maximizes social welfare by inducing zero wasteful effort and ensuring that only productive students are employed. However, high fees allow the monopolist to extract the entire surplus. Competition, as one might expect, redistributes welfare to students. More critically, though, competition reduces total social welfare: While productive students remain fully employed and unproductive students excluded, competition for high types drives the adoption of more informative monitoring structures, increasing wasteful effort and lowering overall welfare. This result reverses the classical economic intuition that competition enhances economic efficiency, potentially raising concerns about pro-competitive policies in the education/testing sector. 

It is worth noting, however, that the sharpness of this conclusion depends on the monopolist's ability to extract the entire surplus through high fees. In what follows, we show that the role of competition is partially restored when students face credit constraints or, equivalently, when schools are subject to a price cap on fees. Indeed, under these market frictions, both monopoly and competition introduce distortions, and the comparison in terms of social welfare depends on the trade-off between job allocation inefficiencies and the cost of wasteful efforts.



\subsection{Extension: Credit Constraint}\label{sec:credit}
So far, we have assumed that students' ability to enroll in schools is unconstrained by their wealth, allowing them to afford arbitrarily high fees. This section assumes students can only afford fees up to some  $K>0$. If $K \geq \theta_H$, the credit constraint is slack, and our preceding analysis holds. Furthermore, since competition depresses fees, our conclusions in the competitive case remain essentially unaffected.\footnote{The only difference is that $f^*=0$ even when competition is less fierce: not only when $n>\frac{1}{\lambda}$ or $n\theta_L>\E\theta$, or 
$-n\theta_L< \theta_H-\theta_L$,  
but also when $n\theta_L>{K}$.} Therefore, we focus on the case where $K \in (0, \theta_H)$ and there is a monopolist school.

In the case of monopoly, students' ability to pay high fees was crucial for the monopolist to extract all the social surplus. While the socially efficient outcome can still arise in the sorting case, in the screening case, the monopolist sacrifices allocative efficiency to maximize profits: in the resulting RPBE, a positive fraction of low types  (with $\theta_L < 0$) enroll and are subsequently employed. Specifically, if $K \geq \E \theta$, the RPBE outcome is unique: (i) the school adopts an uninformative monitoring policy and charges the highest fee $f^* = K$, attracting all high types along with a positive fraction of low types; (ii) all students exert zero effort and receive a wage of $K$ upon graduation. Conversely, if the credit constraint is tighter, $K < \E\theta$, multiple RPBE outcomes arise, characterized by full enrollment and employment (resulting in allocation inefficiency) and, in some cases, positive (wasteful) effort.
\begin{proposition}[Credit Constraint]\label{prop:mon_screenConst} 
Suppose a monopolist school (i.e., $n = 1$) and students can only pay a fee up to $K \in (0, \theta_H)$. 

\begin{enumerate}[(i)]\itemsep0pt
    \item If $K \geq \E\theta$, the RPBE outcome is unique. In the sorting case, the outcome is as described in Proposition \ref{prop:mon_sorting}. In the screening case, the outcome is characterized as follows:
        \begin{enumerate}
            \item[(a)] The school charges the maximum affordable fee, $f^* = K$, and adopts an uninformative monitoring policy, $M^*(e) = 0 \text{ for all } e \in \mathbb{R}_+.
            $
            \item[(b)] All high-type students and a fraction $\alpha_K \in (0, 1]$ of low-type students enroll:         $\psi^*_{\mathbf{p}^*, \theta_H}(i = 1) = 1$ and $\psi^*_{\mathbf{p}^*, \theta_L}(i = 1) = \alpha_K,
            $
            where $\alpha_K$ is given by:
            \begin{align}\label{eq:low enroll prob}
                \frac{\lambda \theta_H + \alpha_K (1 - \lambda) \theta_L}{\lambda + \alpha_K (1 - \lambda)} = K.
            \end{align}
            \item[(c)] Enrolled students exert zero effort: 
           $ \psi^*_{\mathbf{p}^*, \theta}(e = 0 | i = 1) = 1,$ $ \text{for } \theta \in \{\theta_L, \theta_H\}.
           $ 
            \item[(d)] Firms pay all graduates a wage of $K$:             $\omega^*_{\mathbf{p}^*}(m = 0) = K$.
        \end{enumerate}
    \item If $K<\E\theta$, every RPBE outcome still features $f^*=K$ and full enrollment. However, in addition to outcomes where all students exert zero effort, there also exist RPBE outcomes where high-type students exert positive effort.
\end{enumerate}
\end{proposition}
Intuitively, when $K > \E\theta$, the credit constraint binds only in the screening case, as the monopolist's unconstrained optimal policy has $f^* = \theta_H > K$. In this scenario, the monopolist could maintain allocative efficiency, attracting only high types, by allowing them to reveal any effort up to $e^R_K$, where $\theta_H - c(\theta_L, e^R_K) = K$. This policy would ensure that only high types enroll, generating $\lambda K$ revenue. However, if the school limits the maximum signalable effort to $\bar{e} < e^R_K$ while charging $K$, it would also attract a fraction $\alpha_K(\bar{e})$ of low types willing to exert $\bar{e}$ to be pooled with high types. Here, $\alpha_K(\bar{e})$ is defined by:
\[
\frac{\lambda \theta_H + \alpha_K(\bar{e}) (1-\lambda) \theta_L}{\lambda + \alpha_K(\bar{e}) (1-\lambda)} = K + c(\theta_L, \bar{e}).
\]
Since $\alpha_K(\bar{e})$ (and thus profits) increase as $\bar{e}$ decreases, the school optimally sets $\bar{e} = 0$, explaining the result in point (i).

On the other hand, when $K < \E\theta$, multiple RPBEs arise in both the screening (if $\E\theta > 0$) and sorting cases because the school’s maximization problem becomes slack. Once the school attracts all students while charging the maximum fee of $K$, it becomes indifferent to students' effort levels. For instance, the school might offer an uninformative monitoring structure, limiting efforts to zero and leaving students with a utility of $\E\theta - K > 0$. Alternatively, the school could allow students to reveal effort up to $\bar{e}'$ such that $\E\theta - K = c(\theta_L, \bar{e}')$. This would induce an EPBE where all students enroll and exert $\bar{e}'$. While both policies yield the same profit for the school, the latter reduces students’ payoffs to zero, inducing wasteful effort.

\textbf{Welfare Comparison:} When students face credit constraints, or equivalently, when schools face a price cap on fees, both monopoly and competition introduce distortions in the screening case. A monopolist school constrained by a fee cap optimally compensates for lost revenue by expanding enrollment, partially pooling productive and unproductive types, which reduces job allocation efficiency. Whether a planner prefers monopoly or competition depends on the relative magnitude of job allocation inefficiencies versus the cost of wasteful efforts. For instance, if the fee cap is relatively high the monopolist delivers higher social welfare (by continuity). Conversely, if the fee cap is sufficiently low (e.g., resulting in full employment of unproductive types under monopoly) and the cost for high types to exert the Riley effort is relatively low, competition yields higher social welfare.

\section{Discussion}\label{Discussion}
To conclude, we discuss our key assumptions, their role, and possibilities for future investigation.  

Our model assumes that each student's outside option is exogenously determined and normalized to 0. Thus, schools offer not only monitoring structures but also access to employers. This reflects contexts where education is necessary for skill-intensive jobs, while workers without education are restricted to jobs where skills are irrelevant. It also accounts for settings where schooling helps students develop innate abilities, but their exam preparation effort is unproductive. 
Nevertheless, our logic extends beyond these specific scenarios to settings where the outside option is endogenously determined, e.g., when students who opt out earn their (conditional) expected type. Intuitively, while this outside option forces a monopolist to leave students a welfare of $\theta_L$ in the sorting case (when $\theta_L>0$), the resulting monopolist RPBE remains qualitatively unchanged: zero effort is exerted, and all productive types are enrolled and employed. In fact, by allowing students to signal a vanishing level of effort and charging an appropriate fee (converging to $\E\theta - \theta_L$ in the sorting case and to $\theta_H$ in the screening case), schools can prevent equilibria where productive types opt-out and extract their maximum attainable surplus. 
In the competition case, this variation is also unlikely to have a significant impact: fees will still be driven toward 0, and competition for high types will still lead to more wasteful effort. However, the situation becomes more interesting when students' outside options are type-dependent.

Our model assumes that students' types are binary for the economy of exposition. Extending our analysis to multiple types will not alter the main qualitative conclusions. As with binary types, a monopolist school would offer an uninformative monitoring policy (allowing students to signal only a vanishing level effort) and charge a fee of $f^* = \E{\theta|\theta \geq 0}$, inducing no effort and ensuring full employment of all productive types. In the competition case, our logic still holds, with schools' competing policies driving inefficiently high wasteful effort. The main challenge would be characterizing the full set of RPBEs, as schools may compete to attract multiple types.

We limited schools to choosing deterministic monitoring policies and did not allow them to offer menus of fees and monitoring policies to ``screen'' students. These limitations are without loss under monopoly as the school both maximizes and fully extracts the entire surplus.\footnote{Even when students are credit-constrained, such an expansion remains irrelevant, as schools already extract all surplus that can be obtained.} Relaxing the first limitation requires an extension of the refinement. Expanding the policy set to include menus or stochastic monitoring would, at most, introduce vacuous equilibrium multiplicity.  
Even under competition, expanding the available policy set would have minimal impact. Allowing competition through menus would be irrelevant, as students are drawn to a school by the specific combination of fees and monitoring they select, not by the entire menu. Finally, while introducing stochastic monitoring might complicate the characterization of all possible RPBE in competitive settings, it would not alter the result that competition induces inefficiently high effort, thus reducing welfare (the proof of Proposition \ref{prop:competition1} remains valid).

\Xomit{
\textit{Discussion: Intuitive Criterion.} The paper opts to extend the D1 refinement rather than use the more popular Intuitive Criterion. In particular, our game is richer than standard signaling models since we allow for an endogenous and flexible monitoring policy. As a result, the Intuitive Criterion has insufficient ``bite''--- see Appendix \ref{sec:ic} for a detailed discussion. 
}


Finally, our analysis focuses on a set of PBEs that survives the refinement in Definition \ref{df:refine}. This refinement allows us to make sharp predictions about monitoring structures, welfare, and the effects of competition among signaling intermediaries. 
Without such a refinement, the set of PBEs would be vast, as the classical multiplicity of signaling games is compounded by the endogeneity of the monitoring function in our setting. For instance, in the screening case, any monitoring policy could be sustained in a PBE if schools believe that deviations on their part would result in zero enrollment due to firms’ pessimistic beliefs.  
To make progress, we extended the D1 refinement \citep{banks1987equilibrium, cho1987signaling} to our richer setting with endogenous signaling technologies.\footnote{Qualitatively similar results could be obtained by focusing on sequentially stable outcomes \citep{dilme2024sequentially} instead of the $D1$ refinement. Simpler refinements such as the Intuitive Criterion alone  lack sufficient ``bite'' in our setting, leaving multiple equilibrium outcomes even in the monopoly case.}
Our RPBE refinement applies the standard forward induction reasoning to every subgame: given any vector of fees and monitoring policies offered by the school (on- and off-path), if one type incurs a higher cost to send an out-of-equilibrium message than another, the receiver should not believe the message originates from the former type. 

\newpage
\nocite{ali2022sell}
\nocite{lu2019selling}
\nocite{spence1974competitive}
\nocite{Riley1975competitive}
\nocite{dranove2010quality}
\singlespacing

\bibliographystyle{aea}
\bibliography{signaling}

\newpage
\appendix

{
\section{Proof of Theorem \ref{thm:existence}}\label{sec:existence}

In this section, we establish that there exists an EPBE in the subgame following any policy $\mathbf{p} = (f_i, M_i)_{i\in I}$ where, for every $i\in I$, $f_i\leq\theta_H$, and $M_i$ is a deterministic and {right-continuous} function of the student's effort. 
Without loss of generality, let $\mathbb{M} \subset \mathbb{R}$. Further, we assume that $M_i$ is either continuous or gridded (i.e., for any $e \in \mathbb{R}_+$, there exists $k \in \mathbb{N}$ such that $M_i(e) = k \epsilon$, where $\epsilon >0$ can be arbitrarily small).\footnote{The continuous case includes the perfectly informative monitoring policy, i.e., $M_i(e) = e$ for any $e \in \mathbb{R}_+$.}

\subsection{Monopoly} 
Fix the monopolist school's policy to $p = (f, M)$. For any $m \in M(\mathbb{R}_+)$, let $e_m$ denote the minimum effort required to generate it: $e_m = \min_{e \in \mathbb{R}_+ : M(e) = m} e$. Without loss of generality, relabel messages so that $m = e_m$. Denote by $\underline{\psi}$ the strategy that delivers a student the highest payoff between enrolling with zero effort and not enrolling at all; the payoff following $\underline{\psi}$ is therefore at least $\underline{u}=\max\{0, \theta_L - f\}$.
Finally, denote by $e^* \in \mathbb{R}_+$ the maximal effort that the low type might be willing to exert 
\[
\theta_H -c(\theta_L, e^*) -f=  \underline{u},
\]
and define $m^* = \max_{e \in [0,e^*]} M(e)$.
By definition, $
c(\theta_L,e_{m^*})+f \leq \theta_H-\underline{u}< c(\theta_L,e_m)+f$ for any $m>m^*$.
Thus, $e > e_{m^*}$ is strictly dominated for $\theta_L$-students, who would nonetheless prefer exerting $e_{m^*}$ to obtain the wage $\theta_H$ rather than receiving their minimal payoff $\underline{u}$.

\vspace{0.2 in}
We first analyze the case when,
given $p$, it is impossible or too costly for the high type to signal an effort higher than  $e_{m^*}$, i.e., 
either $\{ m \in M(\R):  m > {m}^*\} = \emptyset$ or
\begin{align}\label{ineq:exist_mon_pool}
    \forall m > m^*: \;\;\theta_H-c(\theta_H,e_{m})-f\leq \underline{u}+c(\theta_L,e_{m^*})-c(\theta_H,e_{m^*}).
\end{align}

 Then, following $p$, there exists an EPBE  where:
    \begin{enumerate}
        \item high types enroll and exert $e_{m^*}$; i.e., $\psi_{p,\theta_H}(1,e_{m^*})=1$; low types randomize between $(1,e_{m^*})$, with probability $q\in [0,1]$, and $\underline{\psi}$, with probability $1-q$: 
        \[
q= 
\begin{cases}
    1 &\text{ if }  \bar{w}\leq \E \theta,\\
    \frac{\lambda}{1-\lambda} \left(\frac{ \theta_H - \bar{w}}{\bar{w} - \theta_L }\right) &\text{ otherwise,}
\end{cases}
\]
where $\bar{w} = c(\theta_L, e_{m^*}) +f+ \underline{u} \in [\max\{\theta_L,f\}, \theta_H]$ is the wage that, when paired with effort $e_{m^*}$, yields $\theta_L$-students their minimal payoff $\underline{u}$.

        \item Firms' wage scheme is 
\[
\omega_p(1, m)= 
\begin{cases}
    \max\{\theta_L,0\} &\text{ if }m \in [0, m^*),\\
    \max\{\bar{w},\E \theta \} &\text{ if }m = m^*,\\
    \theta_H  &\text{ if }m > m^*.\\
\end{cases}
\]
The associated belief system is given by
\[
\mu_p(\theta_H|1, m)= 
\begin{cases}
    0 &\text{ if }m \in [0, m^*),\\
    \frac{\lambda }{\lambda + (1-\lambda)q } &\text{ if }m = m^*,\\
    1  &\text{ if }m > m^*.\\
\end{cases}
\]
    \end{enumerate} 
To verify, first note that firms' wage scheme is consistent with their belief system, and on-path beliefs are updated according to Bayes' rule. 

Second, the low type has no incentive to deviate. Indeed, deviations to selecting efforts $e\neq \{0,e_{m^*}\}$ are dominated, so we just need to check the low-type students' payoffs at $(1,e_{m^*})$ and $\underline{\psi}$. If $\bar{w} > \E \theta$, low-type students can optimally randomize (as prescribed) between $(1,e_{m^*})$ and $\underline{\psi}$ as they they both deliver a payoff of $\underline{u}$. Conversely, if $\bar{w} \leq \E \theta$, low-type students can optimally select (as prescribed) $(1,e_{m^*})$ as this allows them them to fully pool with high-type students and obtain a payoff of $\E \theta-c(\theta_L,e_{m^*})-f\geq \bar{w}-c(\theta_L,e_{m^*})-f=\underline{u}$.

Third, high-type students have no incentive to deviate from exerting $e_{m^*}$. Indeed, 
increasing their effort to $\tilde e> e_{m^*}$ is suboptimal. This is trivially true if $M(\tilde e)\leq m^*$.
Further, it is a direct consequence of \eqref{ineq:exist_mon_pool} otherwise; note in fact that  \eqref{ineq:exist_mon_pool} implies that if $M(\tilde e)=m>m^*$, then $\theta_H - c(\theta_H, \tilde{e})-f \leq \underline{u}+c(\theta_L, e_{m^*})-c(\theta_H, e_{m^*})=\bar{w}-f-c(\theta_H, e_{m^*})
\leq \max\{\bar{w},\E \theta \}-c(\theta_H, e_{m^*})-f$. Moreover, since high-type students have a cost advantage in exerting effort and, as we showed, low-type students weakly prefer $(i, e_{m^*})$ to $\underline{\psi}$ (which involves no effort), high-type students must prefer $(i, e_{m^*})$ to $\underline{\psi}$ and, more generally, to any effort level $\tilde{e} < e_{m^*}$.

Finally, we must check that off-path beliefs survive our refinement. First, note that deviation to any higher effort $\tilde e > e_{m^*}$ is unattractive for both types (given that wage is capped by $\theta_H$). Therefore, our refinement does not restrict off-path belief upon observing $m > m^*$; in particular, the belief $\mu(\theta_H|p, m)=1$ for all $m > m^*$ is consistent with our refinement.
Besides, since high-type students have a cost advantage in exerting effort, deviation to any lower effort $\tilde e < e_{m^*}$ is more attractive for low-type students. 
To see this more clearly, for any $m < m^*$, 
\begin{align*}
\omega(p,m^*) - c(\theta_H, e_{m^*}) + c(\theta_H, e_{m}) 
&> \omega(p,m^*)  - c(\theta_L, e_{m^*}) + c(\theta_L, e_{m})
\end{align*}
where the inequality follows from the assumption that $c(\theta,e)$ has strictly decreasing differences. 
Thus, $\mu(\theta_H|p, m)=0$ for all $m < m^*$ is consistent with our refinement.

\vspace{0.2 in}
What remains to discuss is the case where signalling an effort higher than $e_{m^*}$  is not too costly for high-type students; i.e., 
there exists a message $m' > m^*$
such that $\theta_H-c(\theta_H,e_{m'})-f> \underline{u}+c(\theta_L,e_{m^*})-c(\theta_H,e_{m^*}).$ Further, defining $m^*_+ = \inf\left\{m \in M(\mathbb{R}_+): m > m^* \right\},$ we have $m^*_+ \geq m^*$ and
\begin{align}\label{m+}
    \theta_H - c(\theta_H, e_{m^*_+}) - f > \underline{u} + c(\theta_L, e_{m^*}) - c(\theta_H, e_{m^*}).
\end{align}
We show that, in this case, there exists a separating RPBE in the subgame following $p$, where 
    \begin{enumerate}
        \item high-type students exert effort $e=e_{{m}^*_+}$; low-type students choose $\underline{\psi}$;
        
         \item firms' wage scheme is given by
\[
\omega(p, m)= 
\begin{cases}
    \max\{\theta_L,0\} &\text{ if }m \in [0, {m}^*_+),\\
    \theta_H &\text{ if }m \geq {m}^*_+.
\end{cases}
\]
The associated belief system is given by
\[
\mu(\theta_H|p, m)= 
\begin{cases}
   0 &\text{ if }m \in [0, {m}^*_+),\\
   1 &\text{ if }m \geq {m}^*_+.
\end{cases}
\]
    \end{enumerate} 

To verify, first note that firms' wage scheme is consistent with their belief system, and on-path beliefs are updated according to Bayes' rule. Second, low-type students have no incentive to deviate from $\underline{\psi}$ since 
   $\underline{u}=\theta_H-f-c(\theta_L,e^*) \geq \theta_H-f-c(\theta_L,e_{m^*_+})$. Third, high-type students have no incentive to deviate from  $(1,e_{m^*_+})$: (i) $e>e_{m^*_+}$ incurs higher cost effort and delivers the same wage $\theta_H$; (ii) $\underline{\psi}$ results in a lower payoff since, by condition \eqref{m+}, $\underline{u} \leq \underline{u}-c(\theta_H,e_{m^*})+c(\theta_L,e_{m^*}) < \theta_H-c(\theta_H,e_{m^*_+}) - f.$
   
Finally, we check that off-path beliefs survive our refinement. Since deviations to sending any higher message $\tilde m > m^*_+$ are unprofitable for all students, our refinement does not restrict the off-path belief for all $m > m^*_+$. Therefore, we only need to show that any message $\tilde m < m^*_+$ is more attractive for low-type students to send.
Indeed, by definition of ${m^*_+}$, if ${\tilde m<m^*_+}$ then $\tilde m \leq m^*$, i.e., $e_{\tilde m} \leq e_{m^*}$. So, by
condition \eqref{m+} and strictly decreasing differences of $c$ in $(\theta,e)$,
\[\theta_H-c(\theta_H,e_{m^*_+})-f>\underline{u}- c(\theta_H,e_{\tilde m})+c(\theta_L,e_{\tilde m}).\] 
Rearranging terms, we get
\[\underline{u} + c(\theta_L, e_{\tilde m})<\left(\theta_H-c(\theta_H, e_{m^*_+})-f \right)+ c(\theta_H, e_{\tilde m}).\] 
Thus, the prescribed off-path belief $\mu(\theta_H|p, \tilde m)=0$ upon observing any message $\tilde m < m^*_+$ is consistent with our refinement.

\subsection{Competition}
Denote the lowest fee among all schools 
by 
\begin{align*}
    f_{min} = \min_{i\in I} f_i.
\end{align*}
Denote by $\underline{\psi}$ the strategy that provides a student with the highest payoff between (i) enrolling with equal probability in all $i \in I$ such that $f_i = f_{min}$ and exerting no effort, and (ii) not enrolling at all. The payoff following $\underline{\psi}$ is therefore  $\underline{u} = \max\{0, \theta_L - f_{min}\}$.

\vspace{0.1 in}
For any message $m \in M_i(\mathbb{R}_+)$ from school $i \in I$, let $e_{i,m} = \min \{e' \in \mathbb{R}: M_i(e') = m
\}$ denote the minimum effort required to generate it and let $C(\theta, i, m) = c(\theta, e_{i,m}) + f_i$ denote its minimum associated cost for type $\theta \in \Theta$. Without loss of generality, relabel signal $(i,m)$ so that $(i,m) = (i,e_{i,m})$.
Further, for every $i \in I$, denote by $e^*_i \in\R_+$ the maximal effort low-type students are willing to exert in $i$ to gain wage $\theta_H$ rather than getting the reservation payoff $\underline{u}$: 
\[\theta_H - c(\theta_L, e^*_i) - f_i = \underline{u}.\]  
Define $m_i^* = \max_{e \in [0,e_i^*]} M_i(e)$,
and define
\[{m}^* = \max_{i \in I} m_i^*,\,\,\,\,\text{and}\,\,\,\,\,I^* = \argmin_{\{i \in I:m_i^*= {m}^* \} } f_i.\]
Additionally, define,
\begin{align*}
    &\text{``marginal'' signals}: S^* = \{(i, m)\mid i \in I^*, \, m = {m}^*\},\\
    &\text{``high'' signals}: S_+^* =  \{(i, m)\mid i \in I, \, m > m^*\}, \text{ and ,}\\
    &\text{``low'' signals}: S_-^* =  \{(i, m)\mid i \in I^*, \, m < {m}^*\} \bigcup \{(i, m)\mid i \in I\setminus I^*, \,m \leq {m}^*\}.
\end{align*}
Note that $C(\theta, s) = C(\theta, s'')$ for all $s,s''\in S^*.$ Also, by definition, we have 
\begin{align}\label{ineq:exist_comp_compare}
  C(\theta_L, s) \leq \theta_H-\underline{u}  < C(\theta_L, s'), 
\end{align}
for any $s \in S^*$ and any $s' \in S_+^*$.
Thus, as a low-type student, the relative gain from being identified as high type --- instead of obtaining $\underline{u}$ --- can cover the cost to exert effort ${m}^*$ at school $i \in I^*$. 
In contrast, $\underline{\psi}$ dominates sending any ``high'' signal $s'$ that indicates excessively high effort for low-type students. 

We split the discussion depending on the magnitude of this relative gain from being identified as a high type.  

\vspace{0.2 in}
\paragraph{Case 1}: 
If $S_+^* = \emptyset$ or if, for any $s'\in S_+^* $,  
\begin{align}\label{ineq:exist_comp}
    \theta_H  -\underline{u}\leq C(\theta_H,s')-C(\theta_H,s^*)  + C(\theta_L,s^*),
\end{align}
where $s^* \in S^*$, there exists an EPBE ${\mathcal{E}_{\mathbf{p}}}$ in the subgame following $\mathbf{p}$ where 
    \begin{enumerate}
        \item High-type students enroll with equal probabilities in every school $i \in I^*$ and exert effort ${m}^*$, i.e., 
        \[\psi_{\mathbf{p},\theta_H}(i,e) = 
        \begin{cases}
            \frac{1}{|{I}^*| } &\text{ if } i \in I^* \text{ and } e={m}^*,\\
            0 &\text{ otherwise.}
        \end{cases}\]
       \item Low-type students mimic high types ($\psi_{\mathbf{p},\theta_L}=\psi_{\mathbf{p},\theta_H}$) with probability $q\in (0,1]$ and choose $\psi_{\mathbf{p},\theta_L}=\underline{\psi}$ with probability $1-q$, 
       where $q\in [0,1]$ is given by
        \[
q= 
\begin{cases}
    1 &\text{ if } \bar{w} \leq \E \theta,\\
    \frac{\lambda}{1-\lambda} \left(\frac{ \theta_H - \bar{w}}{\bar{w} - \theta_L }\right) &\text{ otherwise,}
\end{cases}
\]
where $\bar{w} = C(\theta_L,s^*) + \underline{u} \in [\max\{\theta_L, f_{min}\}, \theta_H]$ with $s^* \in S^*$, denoting the wage that, when paired with effort ${m}^*$, yields $\theta_L$-students their minimal payoff $\underline{u}$.

        \item Firms' wage scheme is 
\[
\omega_{\mathbf{p}}(s)= 
\begin{cases}
    \max\{\theta_L,0\} &\text{ if } s \in S_-^*,\\
    \max\{\bar{w},\E \theta \} &\text{ if } s \in S^*,\\
    \theta_H  &\text{ if } s \in S_+^*.
\end{cases}
\]
The associated belief system is given by
\[
\mu_{\mathbf{p}}(\theta_H| s)= 
\begin{cases}
    0 &\text{ if } s \in S_-^*,\\
    \frac{\lambda }{\lambda + (1-\lambda)q } &\text{ if } s \in S^*,\\
    1  &\text{ if } s \in S_+^*.
\end{cases}
\]
    \end{enumerate} 

To verify, first note that firms' wage scheme is consistent with their belief system, and on-path beliefs are updated according to Bayes' rule. 

Second, the low type has no incentive to deviate. Indeed, deviations to 
\[\psi_{\mathbf{p},\theta_L}\not\subseteq \Delta \left(\text{supp}(\underline{\psi}) \bigcup \text{supp}( \psi_{\mathbf{p},\theta_H})\right)\] are dominated, so we just need to check the low-type students' payoffs following $ \psi_{\mathbf{p},\theta_H}$ and $\underline{\psi}$. 
If $\bar{w} > \E \theta$, then low type is indifferent between $ \psi_{\mathbf{p},\theta_H}$ and $\underline{\psi}$ as they both deliver payoff $\underline{u}$;
if $\bar{w} \leq \E \theta$, low-type students prefer $ \psi_{\mathbf{p},\theta_H}$ over choosing $\underline{\psi}$ and obtaining $\underline{u}$.

Third, high-type students have no incentive to deviate: (i) generating any ``high'' signal $s \in S^*_+$ yields a lower payoff for high-type students. To see this, 
note
\begin{align*}
   \theta_H  - C(\theta_H,s) \leq \underline{u} + C(\theta_L,s^*) - C(\theta_H,s^*) \leq  \omega_{\mathbf{p}}(s^*) - C(\theta_H,s^*), 
\end{align*}
for any $s\in S^*_+$ and any $s^* \in S^*$, where the first inequality follows from \eqref{ineq:exist_comp}, and the second from the definition of $\bar{w}$. Moreover, since low-type students weakly prefer $\psi_{\mathbf{p},\theta_H}$ to $\underline{\psi}$, then high-type students must prefer  $\psi_{\mathbf{p},\theta_H}$ to $\underline{\psi}$, due to their cost advantage in exerting effort. Thus, high-type students have no strict incentive to generate
any “low” signal $s \in S_-^*$.

Finally, we must check that off-path beliefs survive our refinement. Note that deviation to a ``high'' signal $s \in S^*_+$ (i.e., to signal an effort higher than ${m}^*$) is unattractive for both types (given that wage is capped by $\theta_H$). Hence, our refinement does not restrict the off-path belief upon observing a signal $s\in S^*_+$.  Besides, due to the high-type students' cost advantage, deviations that signal lower effort $e_{s} < {m}^*$ are more attractive to low-type students. Indeed, for every $s=(j,m)\in S^*_-$ such that $e_{s} < {m}^*$, 
\begin{align*}
\omega_{\mathbf{p}}(s^*) - C(\theta_H, s^*) + C(\theta_H, s) &= \omega_{\mathbf{p}}(s^*)  - [c(\theta_H, {m}^*)+ f_i]  + [c(\theta_H, e_{s}) + f_j],\\
&> \omega_{\mathbf{p}}(s^*)  - [c(\theta_L, {m}^*) - c(\theta_L, e_{s})]  -f_i + f_j,\\
&= \omega_{\mathbf{p}}(s^*) - C(\theta_L, s^*) + C(\theta_L, s),
\end{align*}
where $s^*=(i,m^*)\in S^*$ and the inequality follows from the assumption that $c(\theta,e)$ has strictly decreasing differences. Thus $\mu(\theta_H|\mathbf{p}, s)=0$ is consistent with our refinement for all $s\in S^*_-$ such that $e_{s} < {m}^*$. Additionally, our refinement does not restrict the off-path belief when $e_{j,m} = {m}^*$ and $j \notin I^*$. Therefore, the off-path beliefs,  $\mu(\theta_H|\mathbf{p}, s)=1$ for all $s\in S^*_+$ and $\mu(\theta_H|\mathbf{p}, s')=0$ for all $s'\in S^*_-$, are consistent with our refinement. 

\vspace{0.2 in}
\paragraph{Case 2}: 
Suppose now that, given $s^* \in S^*$, there exists a signal $s\in S^*_+$ such that
\begin{align}\label{ineq:exist_comp_sep}
    \theta_H  -\underline{u} > C(\theta_H,s)-C(\theta_H,{s^*})  + C(\theta_L,{s^*}).
\end{align}
Define $\underline{C}^{*}_+ = \inf_{s \in S_+^*} C(\theta_H, s)$ and
\begin{align*}
    \Xomit{\begin{cases}
    \min_{i,m \in M_+^*} C(\theta_H, i,m) &\text{ if } \inf\{e_{i,m}: i,m \in M_+^*, i \in I\} > {m}^*\\
    \min_{i,m \in M_+^* \cup M^*} C(\theta_H, i,m) &\text{ if } \inf\{e_{i,m}: i,m \in M_+^*, i \in I^*\} = {m}^*,\\
    \min_{i,m \in M_+^* \cup {(j,m): \inf\{e_{j,m}: i,m \in M_+^*, i \in I^*\} = {m}^*}} C(\theta_H, i,m) &\text{ otherwise }
    \end{cases}
    }
    \underline{S}^{*}_+ &= 
    \begin{cases}
    \{s \in S_+^*: C(\theta_H, s) = \underline{C}^{*}_+\} &\text{ if }  \{s \in S_+^*: C(\theta_H, s) = \underline{C}^{*}_+\} \neq \emptyset,\\
    \{s\in I\times \mathbb{M}: e_{s} = {m}^* \,\text{ and } \, C(\theta_H, s) = \underline{C}^{*}_+\} &\text{ otherwise. }
    \end{cases}
\end{align*}
Note for all $s \in  \underline{S}^{*}_+ $, $e_{s}\geq {m}^*$.

\Xomit{
Note that $C(\theta_H, i,m) = \underline{C}^{*}_+$ for 
$i,m \in M^*$ only if $\inf_{(j,m)\in S^*_+} e_{j,m} = {m}^*$, that is,
$\theta_H - C(\theta_L, {m}^*) = \underline{u}$. \mn{think more on this}
}

We show that, following $\mathbf{p}$,  there exists a separating EPBE where 
    \begin{enumerate}
    \item high types choose $ \psi_{\theta_H, \mathbf{p}}$ such that
        \[ \psi_{\theta_H, \mathbf{p}}(i,e) = 
        \begin{cases}
            \frac{1}{|\underline{S}^{*}_+| } &\text{ if } (i,M_i(e))  \in  \underline{S}^{*}_+,\\
            0 &\text{ otherwise;}
        \end{cases}\]        
        \item low-type students choose $ \psi_{\theta_L, \mathbf{p}}=\underline{\psi}$;
        
         \item firms' wage scheme is 
         \[
\omega_{\mathbf{p}}(s)= 
\begin{cases}
    \theta_H &\text{ if } s \in \underline{S}_+^* \cup {S}_+^*,\\
    \max\{\theta_L,0\} &\text{ otherwise.}
\end{cases}
\]
The associated belief system is given by
\[
\mu_{\mathbf{p}}(\theta_H| s)= 
\begin{cases}
  1 &\text{ if } s \in \underline{S}_+^* \cup {S}_+^*,\\
 0 &\text{ otherwise.}
\end{cases}
\]
    \end{enumerate}

   To verify, first note that firms' wage scheme and belief system are consistent and, on-path, beliefs are updated according to Bayes' rule. 
   
   Second, low-type students have no incentive to deviate since (i) by \eqref{ineq:exist_comp_compare}, sending a signal $s\in S^*_+$ results in a payoff lower than $\underline{u}$; (ii) signal $s=(i,m) \in S^* \cup  S_{-}^*$ leads to wage $\theta_H$ only if $\inf\{e_{s'}: s' \in S^*_+\} = e_{s} ={m}^*$ and $C(\theta_H, s) = \underline{C}_+^*$. Thus, we can find a sequence $\{e^n\}$ such that $(i,M_i(e^n)) \in S^*_+$ and $\lim_{n \rightarrow \infty} e^n = {m}^* = e_{s}$. By \eqref{ineq:exist_comp_compare}, we have that,  for every $n \in \mathbb{N}$, 
$\theta_H- [c(\theta_L, e^n) + f_i] < \underline{u}$.
   Thus, by the continuity of $c(\theta, e)$ in $e$, we get that $\theta_L$-students (weakly) prefer $\underline{\psi}$ to sending $s$ and earning $\theta_H$:
   \[\theta_H- [c(\theta_L, {m}^*) + f_i] \leq \underline{u}.\]

Third, high-type students have no incentive to deviate from $s \in \underline{S}_+^*$.
Indeed, since $C(\theta_H, s) = \inf_{s' \in S_+^*} C(\theta_H, s')$, sending any $s' \in S_+^* \setminus \underline{S}_+^*$ incurs a higher cost for the same wage, $\theta_H$. Moreover, not enrolling or exerting $e<{m}^*$ also results in lower payoffs as \[\underline{u} \leq  \underline{u} -C(\theta_H,{s^*})  + C(\theta_L,{s^*}) \leq \theta_H-\underline{C}_+^*,\]
for any $s^* \in S^*$, where the second inequality follows from \eqref{ineq:exist_comp_sep} and the definition of $\underline{C}_+^*$.

Finally, we must check that off-path beliefs survive our refinement. Our refinement does not restrict the off-path belief for $s \in  S_+^* \setminus \underline{S}_+^*$, as deviations to sending these signals are unprofitable for all students. Thus, we only need to show that every $(i, e)$ with $s = (i, M_i(e)) \in S_{-}^* \cup (S^* \setminus \underline{S}_+^*)$ is more attractive for low-type students.  Condition \eqref{ineq:exist_comp_sep} implies 
\begin{align*}
    \theta_H  -\underline{u} &> C(\theta_H,s')-C(\theta_H,{s^*})  + C(\theta_L,{s^*})\\
    &= C(\theta_H,s') + c(\theta_L,{m}^*) - c(\theta_H,{m}^*),
\end{align*}
for any $s' \in \underline{S}_+^*$ and $s^* \in S^*$. Then, since $c$ has strictly decreasing differences in $(\theta,e)$ and $e_{s} \leq {m}^*$ for all $s \in S_{-}^* \cup (S^* \setminus \underline{S}_+^*)$, we have
\begin{align*}
    \theta_H - C(\theta_H,s')  &>  \underline{u}+ [c(\theta_L,e_{s}) - c(\theta_H,e_{s})]=
  \underline{u} + C(\theta_L,{s})  - C(\theta_H,{s}),
\end{align*}
and thus 
\[\theta_H - C(\theta_H,{s'}) + C(\theta_H,s) > \underline{u} + C(\theta_L, s).\] 
Thus, $\mu_{\mathbf{p}}(\theta_H| i,m)=0$ for all $s \in S_{-}^* \cup (S^* \setminus \underline{S}_+^*)$ is consistent with our refinement.

\section{Proof of Proposition \ref{prop:mon_screen}}\label{app:proofs-prop-mon-screen}

We first show that for the school to earn $\pi^*=\lambda \theta_H$ in equilibrium, the equilibrium outcome must be the one described in Proposition \ref{prop:mon_screen}.
We then establish that in any RPBE, the school must earn $\pi^*=\lambda \theta_H$. 

\paragraph{Step 1}: \emph{The equilibrium outcome in Proposition \ref{prop:mon_screen} is the only one that generates a school profit of $\pi^* = \lambda \theta_H$, and it is consistent with an RPBE.}

First, note that the maximum total wage the firms may be willing to pay is $\lambda \theta_H$. For the firms to pay $\lambda \theta_H$ in total, low-type students cannot be hired; thus, they must prefer not to enroll.
Thus, the school can attain $\pi^*=\lambda \theta_H$ only if (i) it attracts all and only $\theta_H$-students, (ii) it charges $f^*= \theta_H$, (iii) enrolled students exert no effort (otherwise, they would need to be compensated for that effort) so the school must adopt an uninformative monitoring policy with
$M^*(e) = M^*(0)$ for any $e \in \mathbb{R}_+$ (by the minimality requirement), and (iv) firms pay every graduate a wage of $ \theta_H$ observing the equilibrium message generated.

Consider now a candidate equilibrium that yields the outcome described by Proposition \ref{prop:mon_screen}. The school has no incentive to deviate: it cannot attain higher profits. Students have no incentive to deviate: they are indifferent between enrolling or not, obtaining zero payoffs either way. Firms optimally offer wage $\theta_H$ since only $\theta_H$-students enroll. Thus, the candidate equilibrium is a PBE. Moreover, given the school policy $p^* = (f^*, M^*)$, enrolled students can only generate one signal (no unsent message); thus, the one described is an EPBE outcome following $p^*$. Finally, since EPBE exists following any $p$, but no policy can yield higher school profits, this is also an RPBE outcome.

\paragraph{Step 2}: \emph{The school's RPBE profit must be $\pi^* = \lambda \theta_H$.}

We already argued that $\pi^* \leq \lambda \theta_H$.
Suppose for the sake of contradiction that there exists an RPBE $\mathcal{E}'=(p', \psi', \omega',\mu')$ yielding $\pi' <\pi^* = \lambda \theta_H$. We will show that the school can profitably deviate from $p'$ to  $\hat p =(\hat f, \hat M)$:
\begin{align*}
    &\hat f = \theta_H - c(\theta_H, \epsilon)-\gamma_\epsilon,\\
    &\hat M (e) =
    \begin{cases}
        \hat m_A \quad &\text{if } e \geq \epsilon, \\
        \hat m_B \quad &\text{if } e < \epsilon,
    \end{cases}
\end{align*}
where $\hat{m}_A,  \hat{m}_B \in \mathbb{M}$, and 
$\epsilon,\gamma_\epsilon>0$ are sufficiently small---such that  $ \lambda (\theta_H - c(\theta_H, \epsilon) -\gamma_\epsilon) >\pi'$ and $\gamma_\epsilon<c(\theta_L, \epsilon)-c(\theta_H, \epsilon)$. 

We prove that, following $\hat p$, all $\theta_H$-students enroll ($\psi'_{p,\theta_H}(1) = 1$), yielding a profit of at least $\lambda[\theta_H  - c(\theta_H, \epsilon) -\gamma_\epsilon] > \pi'$. 

First, note that an EPBE exists following any policy $p$.

Now we argue that all high types must enroll following $\hat{p}.$ For the sake of contradiction, suppose the EPBE prescribed by $\mathcal{E}'$ following ${\hat{p}}$ is $\mathcal{E}'_{\hat{p}}=(\hat{p}, \psi'_{\hat{p}}, \omega'_{\hat{p}},\mu'_{\hat{p}})$, such that some high-type students do not enroll; i.e., $\psi'_{\hat{p},\theta_H}(1)<1$. Observe that, even if $\omega'_{\hat{\mathbf{p}}}(1,\hat{m}_A )=  \theta_H$, $\theta_L$-students would prefer not enrolling rather than enrolling and exerting $e=\epsilon$ because $\theta_H-\hat f-c(\theta_L, \epsilon)<\theta_H-\hat f-(c(\theta_H, \epsilon)+\gamma_\epsilon)=0$. Then, no low type enrolls and exerts $e=\epsilon$. We then split the discussion into two cases to show
$\omega'_{\hat{\mathbf{p}}}(\hat{m}_A )=\theta_H$ in $\mathcal{E}'_{\hat{p}}$:
\begin{itemize}
    \item[(i)] If $\theta_H$-students exert $e=\epsilon$ with positive probability following $\hat p$, then by Bayesian updating in the EPBE delivers $\omega'_{\hat{\mathbf{p}}}(\hat{m}_A )=\theta_H$.
    \item[(ii)] If no student selected $e=\epsilon$, then $\omega'_{\hat{\mathbf{p}}}(\hat{m}_A )=\theta_H$ by our refinement. Indeed, since $\theta_H$ and $\theta_L$ would obtain the same equilibrium payoff when not enrolling or not exerting any effort,  $U^{\mathcal{E}'_{\hat{p}}}(\theta_H,\hat p)  = U^{\mathcal{E}'_{\hat{p}}}(\theta_L,\hat p)$. Thus $U^{\mathcal{E}'_{\hat{p}}}(\theta_L,\hat p) + c(\theta_L, \epsilon) > U^{\mathcal{E}'_{\hat{p}}}(\theta_H,\hat p) + c(\theta_H, \epsilon)$. Moreover,  $\theta_H-\hat f- c(\theta_H, \epsilon) = \gamma_\epsilon > 0$.     Thus, by \eqref{ineq:D1}, $\mu'_{\hat{\mathbf{p}}}(\theta_H|1,\hat{m}_A )=1$ and $\omega'_{\hat{\mathbf{p}}}(1,\hat{m}_A )=\theta_H$.
\end{itemize}
But if $\omega'_{\hat{\mathbf{p}}}(\hat{m}_A ) = \theta_H$, then every high-type student would prefer enrolling paying fee $\hat f= \theta_H  - c(\theta_H, \epsilon) -\gamma_{\epsilon}$, exerting $e=\epsilon$ for an effort cost of $c(\theta_H, \epsilon)$, and obtaining wage $\theta_H$, rather than not enrolling and getting the $0$ outside option, a contradiction.

\section{Proof of Proposition \ref{prop:competition1}}
Denote by $f^*_{min}=\min_{i\in I}f^*_i$ the lowest fee charged by schools and let $i^* \in I$ be a school that charges the lowest fee, i.e., $f_{i^*}^*=f_{min}^*$. We will prove that there does not exist any efficient RPBE in both the sorting case and the screening case under competition.

\subsection{Sorting} 

In the sorting case, efficiency (maximum social welfare) requires no effort and full employment. To achieve this, an efficient RPBE $\mathcal{E}^* = (\mathbf{p}^*, \psi^*, \omega^*, \mu^*)$ should be such that: (i) All students enroll, exert no effort, earn $\E \theta$, and attain the same payoff; (ii) Schools set fees $(f_i^*)_{i \in I} \in 
[0, \theta_H]^n$ such that $f^*_{min} \leq \E\theta$ and adopt uninformative monitoring policies with $M_i^*(e) = M_i^*(0)$ for all $i \in I$ and $e \in \mathbb{R}_+$.\footnote{The minimality requirement rules out the possibility that $M_j(e')\neq M_j(0)$ for some $j\in I$ and $e'\in \R_+$ not chosen in equilibrium.} 

\vspace{0.2 in}
First, we show that if $\mathcal{E}^*$ is efficient, then $f_{min}^*=0$.  Suppose, for contradiction, $f_{min}^*>0$.

\paragraph{Case 1}:  
{{{If school $i^*$ does not attract all students,}} then $i^*$ could profitably deviate to $\hat p_{i^*}$:
 \begin{align}
 &\hat f_{i^*}=f_{min}^*-c(\theta_L, \epsilon)-\gamma_{\epsilon}, \nonumber\\
&\hat{M}_{i^*}(e) =  
   \begin{cases}
        \hat m_A \quad &\text{if } e \geq \epsilon,\\
        \hat m_B \quad &\text{if } e < \epsilon,\label{eq:unidevSCR}
    \end{cases} 
\end{align}
where $\hat{m}_A,  \hat{m}_B \in \mathbb{M}$, and $\epsilon, \gamma_{\epsilon}>0$ are sufficiently small. 
To see this, consider $\mathcal{E}^*_{\mathbf{p}'}$, the EPBE prescribed by $\mathcal{E}^*$ in the subgame following $\mathbf{p}'$,  with $p_j' = p^*_j$ for all $j \neq i$ and $p_{i}' = \hat{p}_{i}$. We show that, in $\mathcal{E}^*_{\mathbf{p}'}$, $i^*$ would attract all students: 
\begin{enumerate}[(i)]\itemsep0pt
    \item \textit{All high-type students would enroll in $i^*$.} Note that, under ${\mathbf{p}}'$, $e = 0$ is the only meaningful effort alternative to $e = \epsilon$. Thus, if low-type students choose $(i^*, \epsilon)$ with positive probability, $\theta_H$-students would choose it with probability one due to their cost advantage.
Conversely, if no low-type students select $(i^*, \epsilon)$, the maximal payoff for students enrolling in $j \neq i^*$ would be strictly less than $\E \theta - \hat{f}_{i^*}$, as low-type students would pool with high-type students exerting zero effort. Therefore, either our refinement or Bayes' rule would set $\omega(i^*, \hat{m}_A) = \theta_H > \E \theta$, making high-type students strictly prefer enrolling in $i^*$, as $\E\theta - \hat{f}_{i^*} < \theta_H - \hat{f}_{i^*} - c(\theta_H, \epsilon)$.

\item \textit{Since all high-type students enroll in $i^*$, all low-type students would do the same.} Indeed, enrolling in $j\neq i^*$ would identify students as low-type students, yielding a payoff of at most $\theta_L - f_{min}^*$. In contrast, by mimicking high-type students' choices, their payoff at school $i^*$ is at least $\E \theta - \hat{f}_{i^*} - c(\theta_L,\epsilon) = \E \theta - {f}_{min}^* + \gamma_{\epsilon} > \max\{\theta_L - {f}_{min}^*,0\}$. 
\end{enumerate}
 Thus, all students must enroll in school $i^*$ in every EPBE following $\mathbf{p}'$. Hence, this deviation allows school $i^*$ to ensure higher profits --- covering the entire market with a slightly lower fee. 
Thus, if an efficient RPBE $\mathcal{E}^*$ involves $f_{min}^*> 0$ then it must be that, in $\mathcal{E}^*$, all students enroll in $i^*$. 
}

\paragraph{Case 2}:  If, instead, school $i^*$ attracts all students, it is clear that any competing school $j \neq i^*$ would strictly benefit from unilaterally deviating to $\hat{p}_{j} = \hat{p}_{i^*}$ as defined in \eqref{eq:unidevSCR}.

\vspace{0.2 in}
Thus, if an RPBE $\mathcal{E}^*$ is efficient, it must be that $f_{min}^*= 0$. 

\vspace{0.2 in}
{Second, we show that if an RPBE $\mathcal{E}^*$ is efficient, it cannot have $f_{\text{min}}^* = 0$, leading to the conclusion that no efficient RPBE exists. Suppose, for contradiction, that $f_{i^*}^* = 0$, meaning $i^*$ would make zero profits in $\mathcal{E}^*$. However, $i^*$ could ensure positive profits by unilaterally deviating to 
 \begin{align}
 &\hat f_{i^*} =\gamma_\epsilon, \nonumber \\ 
&\hat{M}_{i^*}(e) =  
   \begin{cases}
        \hat m_A \quad &\text{if } e \geq \epsilon,\\
        \hat m_B \quad &\text{if } e < \epsilon, \label{eq:unidevSCR2}
    \end{cases} 
\end{align}
where $\hat{m}_A,  \hat{m}_B \in \mathbb{M}$, and $\gamma_\epsilon, \epsilon>0$ are sufficiently small---such that $\theta_H-c(\theta_H,\epsilon)-\gamma_\epsilon>\E\theta$.  Indeed, suppose, for contradiction, that no student enrolls in school $i^*$ in the corresponding EPBE $\mathcal{E}^*_{\mathbf{p}'}$. 
  Given $\mathbf{p}'$, no positive effort can be signaled in $j \neq i^*$, and low-type students can mimic high-type students (obtaining the same payoff as $e=0$). Thus, every student payoff in $\mathcal{E}^*_{\mathbf{p}'}$ would be bounded above by $\E \theta$. Thus, given the high type’s cost advantage, our refinement would imply $\omega_{\mathbf{p}'}(i^*,\hat m_A)=\theta_H>\E \theta$. Consequently, high-type students would prefer to switch and enroll in $i^*$, as $\theta_H - c(\theta_H, \epsilon) - \gamma_\epsilon > \mathbb{E} \theta$, leading to a contradiction.
} Thus, $\mathcal{E}^*_{\mathbf{p}'}$ must involve some students enrolling in $i^* \in I$, resulting in positive profits for $i^*$.
Therefore, an efficient RPBE with $f_{min}^* =0$ cannot be sustained.}

\subsection{Screening}
{In the screening case, efficiency requires no enrollment from $\theta_L$, full enrollment from $\theta_H$, and no effort exerted. To achieve this, an efficient RPBE $\mathcal{E}^* = (\mathbf{p}^*, \psi^*, \omega^*, \mu^*)$ should be such that: (i) All $\theta_H$-students enroll, exert no effort, and earn $\theta_H$; (ii) Active schools (those with positive enrollment) charge $f^*_i=\theta_H$---otherwise, $\theta_L$ would mimic high-type students by enrolling and exerting no effort; (iii) All schools adopt uninformative monitoring policies with $M_i^*(e) = M_i^*(0)$ for all $i \in I$ and $e \in \mathbb{R}_+$ (by the RPBE minimality requirement).
We consider two cases to prove that an efficient RPBE does not exist.

\paragraph{Case 1:} Suppose $\mathcal{E}^*$ is such that $f^*_{min} \geq \theta_H$.  There must exist a school $i \in I$ that does not attract all high-type students in $\mathcal{E}^*$ and, as a result, earns a profit $\pi^*_i < \lambda \theta_H$. We show that this school could attract all $\theta_H$-students and attain higher profits by unilaterally deviating to $\hat{p}_{i}=(\hat{f}_{i},\hat{M}_{i})$:
  \begin{align*}
  &\hat f_{i}=\theta_H-c(\theta_H,\epsilon)-\gamma_{\epsilon},\\
&\hat{M_i}(e) =  
   \begin{cases}
        \hat m_A \quad &\text{if } e \geq \epsilon,\\
        \hat m_B \quad &\text{if } e < \epsilon,
    \end{cases} 
\end{align*}
where $\hat{m}_A,  \hat{m}_B \in \mathbb{M}$, and $\epsilon, \gamma_{\epsilon}>0$ are sufficiently small such that 
\[\gamma_{\epsilon} <c(\theta_L,\epsilon) - c(\theta_H,\epsilon), \,\,\, \text{and} \,\,\, \theta_H-c(\theta_H,\epsilon) -\gamma_{\epsilon} > {\pi_i^*}/{\lambda}.\]

Indeed consider $\mathcal{E}^*_{\mathbf{p}'}$, the EPBE prescribed by $\mathcal{E}^*$ in the subgame following $\mathbf{p}'$,  with $p_j' = p^*_j$ for all $j \neq i$ and $p_{i}' = \hat{p}_{i}$.
{ First, note that no $\theta_L$ student would select $(i,e=\epsilon)$ in $\mathcal{E}^*_{\mathbf{p}'}$ as $\theta_H-c(\theta_L,\epsilon)-\hat{f_i}<0$. Thus, we must have $\omega^*_{\mathbf{p}'}(i,\hat m_A)=\theta_H$.
As a result, $\theta_H$-students would strictly prefer $(i,e=\epsilon)$ over enrolling in any other school paying fee $\theta_H$ (in this case, their payoff would be at most 0). Thus, in $\mathcal{E}^*_{\mathbf{p}'}$, all $\theta_H$-students must enroll in $i$. Since $\lambda \hat f_i = \lambda (\theta_H-c(\theta_H,\epsilon) - \gamma_{\epsilon}) > \pi_i^*$, school $i$ has a strict incentive to deviate to the proposed policy, a contradiction. Therefore, such efficient RPBE cannot exist.}

\paragraph{Case 2:} Suppose $f^*_{min} = f^*_{i^*}< \theta_H$. If $\mathcal{E}^*$ is efficient, then no student can enroll in $i^*$ (see (ii) above), implying $\pi_{i^*}^*=0$. However, $i^*$ could ensure positive profits by unilaterally deviating to 
 \begin{align}
 &\hat f_{i^*} =\gamma_\epsilon, \nonumber \\
&\hat{M}_{i^*}(e) =  
   \begin{cases}
        \hat m_A \quad &\text{if } e \geq \epsilon,\\
        \hat m_B \quad &\text{if } e < \epsilon, \label{eq:unidevSCRr}
    \end{cases} 
\end{align}
where $\hat{m}_A,  \hat{m}_B \in \mathbb{M}$, and  $\gamma_\epsilon,\epsilon >0$ are sufficiently small---such that $\theta_H-c(\theta_H,\epsilon)-\gamma_\epsilon>\max\{0,\E\theta\}$.
Indeed, suppose, for contradiction, that no student enrolls in school $i^*$ in the corresponding EPBE $\mathcal{E}^*_{\mathbf{p}'}$. 
 Given $\mathbf{p}'$, no positive effort can be signaled in $j \neq i^*$, and low-type students can mimic high-type students (obtaining the same payoff as $e=0$). Thus, every student payoff in $\mathcal{E}^*_{\mathbf{p}'}$ would be bounded above by $\max\{0,\E \theta\}$. Thus, given the high type’s cost advantage, our refinement would imply $\omega_{\mathbf{p}'}(i,\hat m_A)=\theta_H>\max\{0,\E \theta\}$. Consequently, high-type students would prefer to switch and enroll in $i^*$, as $\theta_H - c(\theta_H, \epsilon) - \gamma_\epsilon > \max\{0,\E \theta\}$, leading to a contradiction.
} Thus, $\mathcal{E}^*_{\mathbf{p}'}$ must involve some students enrolling in $i^* \in I$, resulting in positive profits for school $i^*$.
Therefore, an efficient RPBE with $f_{min}^* <\theta_L$ cannot be sustained.

\section{Proofs of Propositions \ref{prop:R}, \ref{prop:comp_Riley}, \ref{prop:competition_fierce}, and \ref{prop:competition2}} \label{app:proofs-prop-comp}

This section contains proofs for a collection of Propositions from the main text, and some additional results we include for completeness. First, we state the equilibria under the case of Mild Competition (Proposition \ref{prop:competition2}). To avoid duplicating arguments across separate proofs, we then provide a unified proof for Propositions \ref{prop:R}, \ref{prop:comp_Riley}, \ref{prop:competition_fierce}, and \ref{prop:competition2}, highlighting the specific parts addressing each proposition.

\subsection{Mild Competition Among Schools}\label{apd:comp}
In this section, we examine the set of RPBEs that arise under mild competition, i.e., when none of the conditions in Definition \ref{def:fierce} are met. While the RPBE outcomes characterized in Proposition \ref{prop:competition_fierce} remain valid, mild competition also allows for symmetric RPBEs in which schools charge positive fees. As we demonstrate, these symmetric equilibria must be semi-pooling.
\begin{proposition}\label{prop:competition2} 
Under mild competition, the outcome of every symmetric RPBE $\mathcal{E} ^*= (\mathbf{p}^*, \psi^*, \omega^*, \mu ^*)$ falls into one of two categories:
\begin{enumerate}
   \item[I.] The separating RPBE outcome of Proposition \ref{prop:R} (\textbf{Riley outcome}), or

\item[II.] {Semi-pooling}, such that: \begin{enumerate}
   \item[(i)] Schools adopt the following policy $p^* = (f^*, M^*)$:
   \begin{align*}
   &f^* \in \{0\} \cup \left[n\theta_L, \min\left\{\theta_H, \frac{\E\theta}{\lambda n}\right\}\right) &&\text{ if $\theta_L \geq 0$}, \\
   &f^* \in \left[0, \max\left\{\frac{\theta_H + (n - 1) \theta_L}{n}, 0\right\}\right]  &&\text{ if $\theta_L < 0$}, \\
   &M^*(e) =  
   \begin{cases}
        m_A \quad &\text{if } e \geq e_h,\\
        m_B \quad &\text{if } e \in [e_l, e_h),\\
        m_C \quad &\text{otherwise,}
    \end{cases}
\end{align*} 
where $0\leq e_l < e_h$ and $m_A, m_B, m_C 
 \in \mathbb{M}$. In particular, if $f^* = 0$, then $e_h = e^{R}$.

\item[(ii)] All students enroll; low types exert $e_l$; high types exert $e_l$ with probability $q_h \in (0,1)$ and $e_h$ with probability $1-q_h$: $\psi^*_{p^*, \theta_H}(i, e_l)={q_h}/{n}$, $\psi^*_{p^*, \theta_H}(i, e_h)={(1-q_h)}/{n}$,  $\psi^*_{p^*, \theta_L}(i, e_l)={1}/{n}$ 
     for any school $i \in I$.

\item[(iii)] Firms offer the following wage schedule to graduates of any school $i \in I$:
\begin{align*}
   \omega^*_{p^*}(i,m) &= 
   \begin{cases}
        \theta_H \quad &\text{if } m = m_A,\\
        w_l \quad &\text{if } m = m_B,\\
        \max\{\theta_L, 0\} \quad &\text{otherwise,}
    \end{cases}
\end{align*}
 where
\begin{align} \label{eq:mixed wage2}
    w_l = \frac{\lambda q_h \theta_H + (1-\lambda) \theta_L}{\lambda q_h  + 1-\lambda} \in (\max\{\theta_L, 0\}, \theta_H).
\end{align}

\end{enumerate}
\end{enumerate}

\end{proposition}

\bigskip 
\subsection{Proof}
Consider any symmetric RPBE $\mathcal{E}^* = (\mathbf{p}^*, \psi^*,\mu^*, \omega^*)$ with $n\geq 2$ competing schools, where $p_i^*= p^* \equiv (f^*, M^*)$ for all $i\in I$. Let ${e}_h^*$ be the maximal equilibrium effort exerted with positive probability in $\mathcal{E}^*_{\mathbf{p}^*}$, the EPBE following $\mathbf{p}^*$.\footnote{\label{fnt:finite effort}The maximum exists since at most $3$ equilibrium effort levels are exerted with positive probability in a symmetric RPBE. Otherwise, there must be at least two effort levels, $0 \leq e_1 < e_2$, exerted with positive probability by both high- and low-type students, which is impossible since high types must strictly prefer $e_2$ when low-type students are indifferent between $e_1$ and $e_2$.} Note that, due to their cost advantage, $\theta_H$-students must exert ${e}_h^*$ with positive probability in $\mathcal{E}^*_{\mathbf{p}^*}$.

\paragraph{Step 1:}\emph{ If $\theta_L \geq 0$, then $f^* \in \{0\} \cup \big[n\theta_L, \min\{\theta_H,\E\theta/(\lambda n)\}\big)$. If $\theta_L < 0$, then $f^* \in [0, \max\{{[\theta_H + (n - 1) \theta_L]}/{n}, 0\}]$. Further Proposition \ref{prop:comp_Riley} holds, i.e., if competition is fierce (Definition \ref{def:fierce}), then $f^*=0$. } 
\newline Clearly, if $f^* > \theta_H$, no student enrolls as wages cannot exceed $\theta_H$, engendering a strict incentive for schools to deviate. Therefore, we only need to consider the case when $f^* \in [0, \theta_H]$.
Suppose first that  $f^* > 0$. 
Consider the subgame following  $\hat{\mathbf{p}}$ where $\hat p_j=p^*$ for all $j\neq i$ and school $i$ unilaterally deviates to $\hat p_i$ such that:
\begin{align}
    &\hat{f}_i  \in (0, f^*), \nonumber \\
    &\hat{M}_i(e) = e \text{ for every } e \geq 0. \label{eq:M hat}
\end{align}
Define a strictly increasing function $E:\R_{+}\to \R_{+}$ such that 
\[\forall e \in \Re_+: \; 0 < c(\theta_L,E(e))-c(\theta_L,e) < f^* - \hat{f}_i.\] 
We now show that $i$ can attract all high types with this unilateral deviation. Suppose, for contradiction, that in $\mathcal{E}^*_{\hat{\mathbf{p}}}$, a positive mass of high types does not enroll in $i$, i.e., $\psi^*_{\hat{\mathbf{p}}, \theta_H}(i) < 1$.
Denote by $\Hat{{e}}_h \geq 0$ the maximal effort exerted outside school $i$ under  $\mathcal{E}^*_{\hat{\mathbf{p}}}$.\footnote{The maximum exists since (i) there are at most $3(n-1)$ effort levels that are exerted with positive probability in the other $(n-1)$ schools under any policy profile (for the same reason as shown in footnote \ref{fnt:finite effort}); (ii) if no students enroll in schools other than $i$, then, by assumption, there exists a positive mass of high students who chooses the outside option $j = 0$, and thus $\Hat{{e}}_h = 0$.

}  Note that, if $\Hat{{e}}_h = 0$, then we must have 
\begin{align*} U^{\mathcal{E}^*_{\hat{\mathbf{p}}}} (\theta_H) - U^{\mathcal{E}^*_{\hat{\mathbf{p}}}} (\theta_L) = 0 < c(\theta_L, E(0)) - c(\theta_H, E(0)).
\end{align*}
Moreover, if $\Hat{{e}}_h > 0$, then due to their cost advantage, high types must select $(j,\hat{e}_h)$ with positive probability, $\psi^*_{\hat{\mathbf{p}},\theta_H}(j,\hat{e}_h)>0$, for some $j \in I\backslash\{i\}$. Thus
\begin{align*}   U^{\mathcal{E}^*_{\hat{\mathbf{p}}}} (\theta_H) - U^{\mathcal{E}^*_{\hat{\mathbf{p}}}} (\theta_L) \leq  c(\theta_L, \hat{e}_h) - c(\theta_H, \hat{e}_h) < c(\theta_L, E(\hat{e}_h)) - c(\theta_H, E(\hat{e}_h)),
\end{align*}
where the first inequality follows from the incentive compatibility constraint of low types, and the second one from $c(\theta, e)$ being strictly submodular in $(\theta, e)$. As a result:

\begin{enumerate}[(i)]\itemsep0pt
\item If $\psi^*_{\hat{\mathbf{p}}, \theta_L}(i, E(\hat{e}_h)) = 0$, then either Bayes' updating (if $\psi^*_{\hat{\mathbf{p}}, \theta_H}(i, E(\hat{e}_h)) > 0$) or our refinement \eqref{ineq:D1}  (if $\psi^*_{\hat{\mathbf{p}}, \theta_H}(i, E(\hat{e}_h)) = 0$) 
 implies $\mu^*_{\hat{\mathbf{p}}}(\theta_H | i, E(\hat{e}_h)) = 1$, and thus $\omega^*_{\hat{\mathbf{p}}}(i, E(\hat{e}_h)) = \theta_H$.

 However, this contradicts the assumption that $\psi^*_{\hat{\mathbf{p}}, \theta_H}(i) <1$ as high types would strictly prefer $(i, E(\hat{e}_h))$ over $(j,\hat{e}_h)$ where $j \in \{0\} \cup I \backslash \{i\}$:
\begin{align}\label{ineq:high_compare_1}
    \theta_H - \hat f_i - c(\theta_H, E(\hat{e}_h)) > \theta_H - f^* - c(\theta_H, \hat{e}_h) \geq U^{\mathcal{E}^*_{\hat{\mathbf{p}}}} (\theta_H).
\end{align}

\item If $\psi^*_{\hat{\mathbf{p}}, \theta_L}(i, E(\hat{e}_h)) >0$, then high types must strictly prefer $(i, E(\hat{e}_h))$ over $(j,\hat{e}_h)$ 
(since $E(\hat{e}_h) > \hat{e}_h$ and high types have a cost advantage in exerting effort), a contradiction. 
\end{enumerate}
Therefore, following $\hat{\mathbf{p}}$, all high types must enroll in $i$, so $\psi^*_{\hat{\mathbf{p}}, \theta_H}(i) = 1$, resulting in school $i$ earning at least $\lambda \hat{f}_i$. Since all high types enroll in school $i$, students enrolled in any $j \neq i$ would be identified as low-type students and receive a payoff of at most $\max\{0, \theta_L - f^*\}$ under $\mathcal{E}^*_{\hat{\mathbf{p}}}$.

We are now ready to use this characterization to show the following claims.

\begin{claim}
    $f^* < \theta_H.$
\end{claim}
\begin{proof}
Suppose $f^* = \theta_H$. Then, only high types may enroll in schools under $\mathcal{E}^*_{\mathbf{p}^*}$; otherwise, some students' wages would fall below $\theta_H$ by Bayes' rule, making them strictly better off not enrolling, a contradiction. Thus, the maximum profit that school $i$ could achieve under $\mathcal{E}^*_{\mathbf{p}^*}$ is ${\lambda \theta_H}/{n}$.
However, by unilaterally deviating to $(\hat{f}_i, \hat{M}_i)$ with $\hat{f}_i = \theta_H - \epsilon$ for sufficiently small $\epsilon > 0$, school $i$ earns at least $\lambda (\theta_H - \epsilon) > {\lambda \theta_H}/{n}$, a contradiction.
\end{proof}
\begin{claim}
    If ${f^*>0}$ then ${f^* \geq n \theta_L}$.
\end{claim} 
\begin{proof}
This result trivially holds in the screening case ($\theta_L < 0$). In the sorting case ($\theta_L > 0$), suppose, for contradiction, that $f^* \in (0, n\theta_L)$. Then, it is strictly profitable for school $i$ to unilaterally deviate to $\hat p_i$, with $\hat f_i=\min\{\theta_L,f^*\}- \epsilon$ and $\epsilon>0$ sufficiently small, a contradiction. Indeed, under $\mathcal{E}^*_{\hat{\mathbf{p}}}$, 
enrolling in $i$ gives students a payoff of at least $\theta_L - \hat f_i > \max\{0, \theta_L - f^*\}$. Thus, low types also strictly prefer enrolling in school $i$. Therefore, by unilaterally deviating to $(\hat f_i, \hat M_i)$, school $i$ attracts all students and earns $\min\{\theta_L,f^*\}-\epsilon$,  which is strictly more than ${f^*}/{n}$, the maximal profit  school $i$ could obtain under $\mathcal{E}^*_{{\mathbf{p}^*}}$.
\end{proof}

\begin{claim}
    If $n>{1}/{\lambda}$, then $f^*=0$.
\end{claim}
\begin{proof}
Suppose $f^*>0$. If $n>{1}/{\lambda}$, the maximum profit $i$ could achieve under $\mathcal{E}^*_{\mathbf{p}^*}$ is ${f^*}/{n}<\lambda f^*$. However, by unilaterally deviating to $(\hat f_i, \hat M_i)$, with $\hat f_i = f^*-\epsilon$ and $\epsilon>0$ arbitrarily small, 
school $i$ attracts all high types, earning at least $\lambda \hat f_i = \lambda (f^* - \epsilon)>{f^*}/{n}$, a contradiction. 
\end{proof}

\begin{claim}
    If $n \theta_L > {\E \theta}$, then $f^* = 0$.
\end{claim}
\begin{proof}
    Suppose $n \theta_L > {\E \theta}$, which can only hold if $\theta_L> 0$ (i.e., the sorting case). As shown in Claim 2, if $f^* > 0$, then $f^* \geq n \theta_L$ and school $i$ can secure a profit arbitrarily close to $\theta_L$ by unilaterally deviating to $(\hat f_i, \hat M_i)$ with $\hat{f}_i = \theta_L - \epsilon$ for $\epsilon > 0$ arbitrarily small. Since ${\E \theta}/{n}$ is the maximum profit a school can earn in a symmetric RPBE when $\theta_L> 0$, this deviation is profitable if $\theta_L > {\E \theta}/{n}$.\footnote{Recall that,  in the sorting case, $\E \theta$ is the maximum social surplus.} Thus, $n \theta_L > {\E \theta}$ implies $f^* = 0$.
\end{proof}

\begin{claim}
If $\theta_L < 0$, $f^* \leq \max\{{[\theta_H + (n - 1) \theta_L]}/{n}, 0\}$.
\end{claim}
\begin{proof} 
Let $q_l \in[0,1]$ be the proportion of low-type students who enroll in schools under $\mathcal{E}^*$. 
Note that school $i \in I$ would profitably deviate from $p^*$ to $\hat{p}_i = (\hat f_i, \hat M_i)$, with $\hat f_i = f^*-\epsilon$ and $\epsilon>0$ arbitrarily small,  if its profits were $\pi^* < \lambda f^* $. Since $\pi^* \leq [\lambda +(1-\lambda) q_l]f^*/n$, to avoid profitable deviations,  we must have $[\lambda +(1-\lambda) q_l]f^*/n\geq \lambda f^*$, i.e., $ q_l \geq  {(n-1)\lambda}/{(1-\lambda)}$.
Moreover, since $f^*$ cannot exceed $\max\{ \frac{\lambda \theta_H+(1-\lambda)q_l\theta_L}{\lambda +(1-\lambda)q_l}, 0\}$, we have 
   $ \pi^* \leq \max\left\{ \frac{\lambda \theta_H+(1-\lambda)q_l\theta_L}{\lambda +(1-\lambda)q_l} 
  , 0 \right\}\frac{\lambda +(1-\lambda)q_l}{n} \leq  \max\left\{\frac{\lambda[\theta_H+(n-1)\theta_L]}{n}, 0 \right\},$
 where the second inequality follows from $\theta_L<0$ and $ q_l \geq  {(n-1)\lambda}/{(1-\lambda)}$. 
 
 Thus, to prevent profitable deviations, we need $\max\left\{{\lambda[\theta_H+(n-1)\theta_L]}/{n}, 0 \right\}\geq \lambda f^*$, i.e., $f^* \leq \max\{{[\theta_H + (n - 1)\theta_L]}/{n}, 0 \}$. 
 \end{proof}

\begin{remark}
    Note that the previous claim implies that if  $ -(n - 1)\theta_L \geq \theta_H$, then $f^* = 0$.    
\end{remark}
 
\begin{claim}
If $\theta_L \geq 0$, $f^* < \E\theta/(\lambda n)$.
\end{claim}
\begin{proof}
Note that $\pi^*< \E\theta/n$ if $\theta_L \geq 0$: indeed $\E \theta$ is the maximum social surplus and can only be obtained in a fully pooling RPBE (which we will prove cannot exist in Step 4). Moreover, the unilateral deviation to $\hat{p}_i = (\hat f_i, \hat M_i)$ is profitable, with $\hat f_i = f^*-\epsilon$ and $\epsilon>0$ arbitrarily small, unless $\pi^*\geq \lambda f^*$.    
\end{proof}

\noindent These claims taken together conclude Step 1.

\paragraph{Step 2:}\emph{ If a symmetric RPBE $\mathcal{E}^*$ involves $f^* > 0$ (as in  [Proposition \ref{prop:competition2}]), then its outcome is semi-pooling, with full enrollment and high types exerting the same effort as low-type students, $e_l\geq 0$, with probability $q_h \in (0,1)$ and a higher effort $e_h>e_l$ with probability $1 - q_h$.
} 
\medskip

\noindent
Suppose a symmetric RPBE $\mathcal{E}^*$ involves $f^* > 0$. Then, as shown above, $f^* > \theta_L$, and students strictly prefer not enrolling over paying $f^*$ and being identified as low types. Thus, there is no $(i, e) \in I \times \mathbb{R}_+$ such that $\psi^*_{\mathbf{p}^*, \theta_L}(i, e) > 0 = \psi^*_{\mathbf{p}^*, \theta_H}(i, e)$. 
This leaves only three possible classes of symmetric RPBE:
\begin{enumerate}[(i)]\itemsep0pt
    \item Full separation: Low types do not enroll, high types enroll and exert $e_h\geq 0$;
    \item Pooling: All enrolled students exert the same effort $e_l\geq 0$;
    \item Semi-pooling: Among enrolled students, low types exert $e_l\geq 0$, high types exert $e_l\geq 0$ with probability $q_h \in (0,1)$ and $e_h>e_l$ with probability $1-q_h$.
\end{enumerate}

Note that $\mathcal{E}^*$ cannot belong to the first class. 
Indeed, as shown, by deviating to $\hat{p}_i$ with $\hat{f}_i = f^* - \epsilon$ for any small $\epsilon > 0$, school $i$ can secure a profit of at least $\lambda (f^* - \epsilon)$, which strictly exceeds its maximal RPBE profit of $\lambda f^* / n$.
 
$\mathcal{E}^*$ cannot belong to the second class. Indeed, assume that, in $\mathcal{E}^*$, low types enroll and exert $e_l \geq 0$ with probability $q_l \in (0,1]$ and do not enroll otherwise; i.e. $\psi^*_{\mathbf{p}^*, \theta_L}(j=0,e=0)=1-q_l$ and $\psi^*_{\mathbf{p}^*, \theta_L}(i,e_l)=q_l/n$ for every $i\in I$. 
For contradiction, suppose all (enrolled) high types also exert $e_l$ in $\mathcal{E}^*_{\mathbf{p}^*}$, i.e., $\psi^*_{\mathbf{p}^*, \theta_H}(j=0,0)=1-b$ for $b\in(0,1]$ and $\psi^*_{\mathbf{p}^*, \theta_H}(i,e_l)=b/n$ for every $i\in I$. Thus, \[\frac{\frac{\lambda}{n}\theta_H+(1-\lambda)\frac{q_l}{n}\theta_L}{{\frac{\lambda}{n}+(1-\lambda)\frac{q_l}{n}}} \geq \omega_{\mathbf{p}^*}^*(i,M^*(e_l))\geq f^*+c(\theta_L,e_l).\]
Consider school $i$'s unilateral deviation to $\Tilde{{p}_i}$ such that
\begin{align}
    &\tilde{f}_i =
f^* -[c(\theta_L,e_l +\epsilon)-c(\theta_L,e_l)] -\gamma_{\epsilon}, \nonumber\\
    &\tilde{M}_i(e) =
    \begin{cases}
     m_A \quad &\text{ if } e\geq e_l +\epsilon,\\
    e \quad &\text{ if } e< e_l +\epsilon,
    \end{cases} \label{eq:M tilde}
\end{align}
where $m_A \in \mathbb{M}$ and $\epsilon, \gamma_{\epsilon} >0$ are arbitrarily small. By similar arguments as in Step 1, school $i$ attracts all high types in the resulting EPBE $\mathcal{E}^*_{\Tilde{\mathbf{p}}}$. Additionally, more than a fraction ${q_l}$ low types must enroll in $i$ under $\mathcal{E}^*_{\tilde{\mathbf{p}}}$. Indeed, if a fraction $\alpha\leq q_l$ of low types enrolls in $i$ under $\mathcal{E}^*_{\tilde{\mathbf{p}}}$, low types could earn, by mimicking the high types, a payoff (weakly) greater than $\frac{\lambda\theta_H+(1-\lambda)q_l\theta_L}{{\lambda+(1-\lambda)q_l}}- c(\theta_L,e_l+\epsilon)-\tilde{f}_i> \omega_{\mathbf{p}^*}^*(i, M^*(e_l))- f^*-c(\theta_L,e_l)\geq 0$. Hence, by deviating to $\tilde{p}_i$, school $i$ can earn at least $[\lambda + (1-\lambda) q_l](f^*-\epsilon')$ where $\epsilon' = c(\theta_L,e_l +\epsilon)-c(\theta_L,e_l)+\gamma_{\epsilon}$. This profit is strictly higher than the one under $\mathcal{E}^*_{{\mathbf{p}}^*}$, $[\lambda + (1-\lambda) q_l] f^* /n$, a contradiction. 

Finally, we show that the third class of RPBE requires full enrollment. Indeed, let  $q_l\in(0,1]$ and  $q_h\in(0,1)$ 
be such that $\psi^*_{\mathbf{p}^*, \theta_L}(0,0)=1-q_l$, $\psi^*_{\mathbf{p}^*, \theta_L}(i,e_l)=q_l/n$, $\psi^*_{\mathbf{p}^*, \theta_H}(i,e_l)=q_h/n$, and $\psi^*_{\mathbf{p}^*, \theta_L}(i,e_h)=(1-q_h)/n$ for every $i\in I$, with $e_h>e_l\geq 0$.\footnote{$q_l = 0$ or $q_h = 0$ corresponds to RPBE type (i); $q_h = 1$ corresponds to RPBE type (ii). We have already ruled out both.} 
Suppose, for contradiction, that some students do not enroll in $\mathcal{E}^*$. Then low types must be indifferent between enrolling and not, $ U^{\mathcal{E}^*_{{\mathbf{p}^*}}}(\theta_L)=0$. 
Let $w_l^*=\omega_{\mathbf{p}^*}^*(i,M^*(e_l))$ for any $i \in I$.
The payoffs of low and high types must be, respectively:
\begin{align}
U^{\mathcal{E}^*_{\mathbf{p}^*}}(\theta_L) &= w_l^* - c(\theta_L, e_l)  - f^* = 0, \label{eqn:payoff_low_semi} \\
U^{\mathcal{E}^*_{\mathbf{p}^*}}(\theta_H) &= \theta_H - c(\theta_H, e_h) - f^* =
    w_l^* - c(\theta_H, e_l)  - f^* = c(\theta_L, e_l) - c(\theta_H, e_l), \label{eqn:payoff_high_semi}
\end{align}
where the last equality in \eqref{eqn:payoff_high_semi}  follows from \eqref{eqn:payoff_low_semi}.

We show that school $i$ benefits from unilaterally deviating to $(\tilde f_i,\tilde M_i)$ above. Denote by $\tilde{e}_h \in \{0, e_l, e_h\}$ the maximal effort exerted in schools other than $i$ in $\mathcal{E}^*_{\tilde{\mathbf{p}}}$. 
If $\psi^*_{\tilde{\mathbf{p}}, \theta_L}(j, e) > 0$ for some $j\in I\setminus \{i\}$ and $e\geq 0$, then $\psi^*_{\tilde{\mathbf{p}}, \theta_H}(j, e) > 0$, as students would prefer not enrolling over being identified as $\theta_L$ (recall that $f^*>\theta_L$). 

We first show all high types must enroll in school $i$ following $\tilde p$, i.e., $\psi^*_{\tilde{\mathbf{p}}, \theta_H}(i) =1$. If $\tilde{e}_h \in \{0, e_l\}$, we can use arguments similar to those in Step 1 to conclude $\psi^*_{\tilde{\mathbf{p}},\theta_H}(i)=1$. Next, we focus on the case when $\tilde{e}_h = e_h$ and show, by contradiction, this case cannot arise in $\mathcal{E}^*_{\tilde{\mathbf{p}}}$. 
Given that $\tilde{e}_h = e_h$, we must have $\psi^*_{\tilde{\mathbf{p}}, \theta_H}(j, e_h) > 0$ for some $j\in I\setminus \{i\}$. Hence, using \eqref{eqn:payoff_high_semi},
\begin{align}
    U^{\mathcal{E}^*_{\tilde {\mathbf{p}}}}(\theta_H) - U^{\mathcal{E}^*_{\tilde {\mathbf{p}}}}(\theta_L) \leq  U^{\mathcal{E}^*_{\tilde {\mathbf{p}}}}(\theta_H) 
    &\leq \theta_H - c(\theta_H, e_h)  - f^*, \nonumber \\
    &= c(\theta_L, e_l) - c(\theta_H, e_l), \nonumber \\
    &< c(\theta_L, e_l+ \epsilon) - c(\theta_H, e_l + \epsilon). \label{ineq:D1_semi}
\end{align}
\begin{enumerate}[(i)]\itemsep0pt
\item If $\psi^*_{\tilde{\mathbf{p}}, \theta_L}(i, e_l + \epsilon) = 0$, no low-type student would generate signal $(i,m_A)$. Then either Bayes' updating (if $\psi^*_{\tilde{\mathbf{p}}, \theta_H}(i, e_l + \epsilon) > 0$) or, by \eqref{ineq:D1_semi}, our refinement (if $\psi^*_{\tilde{\mathbf{p}}, \theta_H}(i, e_l + \epsilon) = 0$) would imply that firms' belief upon observing $(i, m_A)$ must be $\mu^*_{\tilde{\mathbf{p}}}(\theta_H | i, m_A) = 1$, and thus $\omega^*_{\tilde{\mathbf{p}}}(i, m_A) = \theta_H$. However,  this contradicts the assumption that $\psi^*_{\tilde{\mathbf{p}}, \theta_H}(i) <1$ as high types would strictly prefer $(i, e_l + \epsilon)$ over $(j,{e}_h)$ where $j \in \{0\} \cup I \backslash \{i\}$:
\begin{align*}
    \theta_H - \tilde{f}_i - c(\theta_H, e_l+ \epsilon) > \theta_H - f^*- c(\theta_H, e_h) \geq U^{\mathcal{E}^*_{\tilde {\mathbf{p}}}} (\theta_H).
\end{align*}

\item If, instead, $\psi^*_{\tilde{\mathbf{p}}, \theta_L}(i, e_l + \epsilon) > 0$, then high types must strictly prefer $(i, e_l + \epsilon)$ over $(j,{e}_h)$. Indeed, by individual rationality, we must have
\begin{align*}
    U^{\mathcal{E}^*_{\tilde {\mathbf{p}}}}(\theta_L) = \omega^*_{\tilde{\mathbf{p}}}(i, m_A)  - c(\theta_L, e_l+\epsilon)  - \tilde{f}_i \geq 0.
\end{align*}
Thus, high types must strictly prefer $(i,e_l+\epsilon)$ over $(j,{e}_h)$ with  $j \neq i$ since
\begin{align*}
    \omega^*_{\tilde{\mathbf{p}}}(i, m_A)  - c(\theta_H, e_l+\epsilon)  - \tilde{f}_i &=  U^{\mathcal{E}^*_{\tilde {\mathbf{p}}}}(\theta_L) + [c(\theta_L, e_l+\epsilon) - c(\theta_H, e_l+\epsilon)]\\
    &\geq c(\theta_L, e_l+\epsilon) - c(\theta_H, e_l+\epsilon)\\
    &>\theta_H - c(\theta_H, e_h)  - f^*\geq U^{\mathcal{E}^*_{\tilde {\mathbf{p}}}}(\theta_H), 
\end{align*} 
where the last two inequalities follow from \eqref{ineq:D1_semi} and the fact that $\psi^*_{\tilde{\mathbf{p}}, \theta_H}(j, e_h) > 0$ for some $j\in \{0\} \cup I\setminus \{i\}$. This contradicts the assumption that $\psi^*_{\tilde{\mathbf{p}}, \theta_H}(i) <1$. \end{enumerate}
Thus, all high types must enroll in school $i$ in $\mathcal{E}^*_{\tilde {\mathbf{p}}}$; i.e., $\psi^*_{\tilde{\mathbf{p}}, \theta_H}(i) = 1$, a contradiction to $\psi^*_{\tilde{\mathbf{p}}, \theta_H}(j, e_h) > 0$ for some $j\in I\setminus \{i\}$. 

Combined with $f^*>\theta_L$, this implies that low types' maximal payoff outside school $i$ is $0$. Moreover, under $\mathcal{E}^*_{\tilde{\mathbf{p}}}$, the cost for low types to mimic high types is at most $c(\theta_L, e_l + \epsilon) + \tilde{f}_i$, which is strictly lower than $c(\theta_L, e_l) + f^*$, the cost they incurred to semi-pool with high types under $\mathcal{E}^*_{\mathbf{p}^*}$.\footnote{Recall that no $e>e_l+\epsilon$ can be signaled in $i$ after the deviation.} Thus, the portion of low types enrolling in $i$  under $\mathcal{E}^*_{\tilde {\mathbf{p}}}$ is strictly greater than $q_l$, their portion under $\mathcal{E}^*_{\mathbf{p}^*}$. As a result, the deviation is strictly profitable for school $i$, leading to a contradiction to partial enrollment.

\paragraph{Step 3:}\emph{ Separating RPBE always exists and induces the Riley outcome \textbf{[Proposition \ref{prop:R}]}.}\\ If a symmetric RPBE $\mathcal{E}^*$ features full separation, we must have: 
\begin{enumerate}[(i)]\itemsep0pt
    \item $f^* = 0=\pi^*$, as $f^* > 0$ corresponds only to semi-pooling RPBE outcomes (Step 2), 
    \item $M^*$ such that:
\begin{align}\label{eq:M sep}
    {M}^*(e) =  
   \begin{cases}
         m_A \quad &\text{if } e \geq e_h,\\
        m_B \quad &\text{otherwise, }
    \end{cases}
\end{align} 
for some threshold effort $e_h > 0$ (by the RPBE minimality requirement). 
\item Students strategies: $\psi^*_{{\mathbf{p}^*},\theta_H}(i,e_h)=1/n$ and $\psi^*_{{\mathbf{p}^*},\theta_L}(i,e)=0$ for all $i\in I$, $e\geq e_h$.
\end{enumerate}

We then show $e_h$ must equal the Riley effort, $e_h = e^R$. First, since by definition $c(\theta_L, e^{R}) = \theta_H - \max\{\theta_L, 0\}$, if low types could earn a wage of $\theta_H$ by paying $f^* = 0$ and exerting $e < e^R$, they would optimally do so; hence, $e_h \geq e^R$. Second, we show that if $e_h> e^{R}$, then school $i$ can earn strictly positive profits by unilaterally deviating to $\hat p_i = (\hat{f}_i, \hat{M}_i)$ where $\hat{f}_i = \gamma_{\epsilon}$ with $\gamma_{\epsilon} > 0$ arbitrarily small, and $\hat{M}_i(e)=e$ for all $e\geq 0$. 
Suppose for contradiction, no students enroll in school $i$ under $\mathcal{E}^*_{\hat{\mathbf{p}}}$. Denote by $\hat e_h \in \{0, e_h\}$ 
the maximal effort exerted by high types outside $i$ under $\mathcal{E}^*_{\hat{\mathbf{p}}}$.  Since $\hat f _i>0$, selecting $(i,e^{R})$ would result in a negative payoff for low types, implying  $\psi^*_{\tilde{\mathbf{p}},\theta_L}(i,e^{R})=0$. 
Hence,  either Bayes' rule (if $\psi^*_{\hat{\mathbf{p}}, \theta_H}(i, e^R) > 0$) or our refinement (if $\psi^*_{\hat{\mathbf{p}}, \theta_H}(i, e^R) = 0$) 
 would imply $\mu^*_{\hat{\mathbf{p}}}(\theta_H | i, e^R) = 1$, and thus $\omega^*_{\hat{\mathbf{p}}}(i, e^R) = \theta_H$. 
 Thus, if $ e_h> e^{R}$, high types would strictly prefer $(i,e^{R})$  over $(j,e_h)$ for all $j\in I$, contradicting $\hat e_h =  e_h$.

If, instead, $\hat e_h = 0$, then $U^{\mathcal{E}^*_{\hat{\mathbf{p}}}} (\theta_H) = U^{\mathcal{E}^*_{\hat{\mathbf{p}}}} (\theta_L) = \max\{\E\theta, 0\}.$ 
Therefore,
\begin{align*}
     U^{\mathcal{E}^*_{\hat{\mathbf{p}}}} (\theta_H) - U^{\mathcal{E}^*_{\hat{\mathbf{p}}}} (\theta_L) = 0 <  c(\theta_L, \epsilon) - c(\theta_H, \epsilon),
\end{align*}
for any $\epsilon>0$. If no student enrolls in $i$ in  $\mathcal{E}^*_{\hat{\mathbf{p}}}$, $(i, \epsilon)$ is an unsent signal. Then, by our refinement, $\mu^*_{\hat{\mathbf{p}}}(\theta_H | i, \epsilon) = 1$, and thus $\omega^*_{\hat{\mathbf{p}}}(i, \epsilon) = \theta_H$ for every $\epsilon>0$ sufficiently small. However, this implies high types would strictly prefer $(i,\epsilon)$ over $(j,0)$ for all $j\neq i$ since
\[\theta_H - c(\theta_H, \epsilon) - \gamma_{\epsilon} > \max\{\E\theta, 0\} = U^{\mathcal{E}^*_{\hat{\mathbf{p}}}} (\theta_H), \]
contradicting zero enrollment in school $i$.

Hence, whenever $e_h \neq e^R$, school $i$ can profitably deviate. Thus, in a separating RPBE $\mathcal{E}^*$, all high types must exert $e_h = e^R$, earning a wage of $\theta_H$ and all low types must exert $e=0$ and earn $\max\{0,\theta_L\}$; i.e., the Riley outcome.

Finally, we verify there exists a separating RPBE $\mathcal{E}^*$ such that:

\begin{enumerate}[(i)]\itemsep0pt
\item Every school $i \in I$ offers $p^* = (f^*, M^*)$ where $f^*= 0$ and $M^*$ is given by \eqref{eq:M sep} with $e_h = e^{R}$.
\item High types select strategy $\psi^*_{\theta_H}$ such that $\psi^*_{{\mathbf{p}^*}, \theta_H}(i, e^R) =1/n$ for all $i \in I$. 
\item Low types select strategy $\psi^*_{\theta_L}$ such that $\psi^*_{{\mathbf{p}^*}, \theta_L}(i, 0) =1/n$ for all $i \in I$ in the sorting case (with $\theta_L> 0$), and $\psi^*_{{\mathbf{p}^*}, \theta_L}(i) =0$ for all $i \in I$ in the screening case (with $\theta_L<0$).
\item Firms wage schedule $\omega^*$ is such that  $\omega^*_{\mathbf{p}^*}(i, m_A) =\theta_H$ and $\omega^*_{\mathbf{p}^*}(i, m_B) =\max\{\theta_L,0\}$ for all  $i \in I$. 
\item For any $\hat{\mathbf{p}}$ such that $\hat{p}_i \neq p^*$ for a school $i\in I$ and $\hat{p}_j=p^*$ for all $j\in I\setminus \{i\}$: 
\begin{enumerate}[(a)]\itemsep0pt
    \item high-type students' strategy is  $\psi^*_{\hat{\mathbf{p}}, \theta_H}(j, e^R) =1/(n-1)$; 
    \item low-type students' strategy is $\psi^*_{\hat{\mathbf{p}}, \theta_L}(j, 0) =1/(n-1)$ in the sorting case or, $\psi^*_{{\hat{\mathbf{p}}}, \theta_L}(i') =0$ for all $i' \in I$ in the screening case;
    \item firms offer $\omega^*_{\hat{\mathbf{p}}}(j, m_A) = \theta_H$, and $\omega^*_{\hat{\mathbf{p}}}(j, m_B) =\omega^*_{\hat{\mathbf{p}}}(i, m) =\max\{\theta_L,0\}$ for any $m \in \hat{M}_i(\mathbb{R}_+)$.
\end{enumerate}
\item $\mathcal{E}^*_{{\mathbf{p}}}$ is an EPBE for every policy profile $\mathbf{p}$.\footnote{Recall that the existence of an EPBE for any policy profile is established in Theorem \ref{thm:existence}.}
\end{enumerate}
Indeed, no student is incentivized to deviate, and on-path wages coincide with the expected type. Besides, schools cannot make positive profits by any unilateral deviation: with a positive fee, (i) providing signals that require a higher effort $e \geq e^{R}$ attracts nobody, and (ii) providing any other signals leads to wage $\max\{\theta_L,0\}$ and attracts nobody. Finally,  all off-path beliefs satisfy our refinement.

\paragraph{Step 4:}\emph{ There exists no fully pooling RPBE.}\\ Suppose, for contradiction, that a symmetric RPBE $\mathcal{E}^*$ satisfies $\psi^*_{\mathbf{p}^*, \theta_L} = \psi^*_{\mathbf{p}^*, \theta_H}$. Then:
\begin{enumerate}[(i)]\itemsep0pt
    \item the RPBE policy $p^*$ must be such that $f^* = 0=\pi^*$ (by Steps 1 and 2) and, 
    \begin{align}
    M^*(e) =  
   \begin{cases}
         m_A & \text{if } e \geq e_l, \\
         m_B & \text{otherwise,}
    \end{cases}
\end{align}
for some threshold effort $e_l \geq 0$ (by the RPBE minimality requirement). 
\item In $\mathcal{E}^*_{\mathbf{p}^*}$, all enrolled students exert $e_l$.\footnote{High and low types can pool at most on one effort.}
\end{enumerate}  
We show school $i$ can profitably attract a positive mass of students by unilaterally deviating to $\hat p_i = (\hat{f}_i, \hat{M}_i)$ where $\hat{f}_i = \gamma_{\epsilon}> 0$ sufficiently small, and $\hat{M}_i(e)=e$ for all $e\geq 0$. 

Suppose, for contradiction, no students enroll in $i$, $\psi^*_{{\hat{\mathbf{p}}}, \theta_H}(i)=\psi^*_{{\hat{\mathbf{p}}}, \theta_L}(i) =0$ for $\hat{\mathbf{p}}$ such that $\hat{p}_i\neq p^*$ and $\hat{p}_j=p^*$ for all $j\in I\setminus \{i\}$. 
Denote by $\hat e_h \in \{0, e_l\}$ the minimal effort exerted by high types outside school $i$, in  $\mathcal{E}^*_{\hat{\mathbf{p}}}$. Define $\hat w_h:=\max_{j\neq i}\omega_{{\hat{\mathbf{p}}}}(j,\hat e_h) \geq 0$ the maximal wage associated with $\hat e_h$ at a school $j \neq i$. First, note that $\hat w_h < \theta_H$: indeed we have $\theta_h-f^*-c(\theta_L, \hat e_h)>\max\{0,\theta_L-f^*\}$, implying low types would optimally pool with high types on $\hat e_h$ at school $j\neq i$ to earn $\hat w_h$, if $\hat w_h\geq \theta_H$.

Second,  by low types' incentive compatibility constraint, 
\[0 \leq U^{\mathcal{E}^*_{\hat{\mathbf{p}}}} (\theta_H)- U^{\mathcal{E}^*_{\hat{\mathbf{p}}}} (\theta_L) \leq c(\theta_L, \hat e_h) - c(\theta_H, \hat e_h)  < c(\theta_L, \hat e_h+ \epsilon) - c(\theta_H, \hat e_h+ \epsilon),\]
for any $\epsilon>0$ sufficiently small. If no students enroll in school $i$ under $\mathcal{E}^*_{\hat{\mathbf{p}}}$\, , $(i, \hat e_h+ \epsilon)$ is an unsent signal. Then by our refinement, firms' belief would be such that $\mu^*_{\hat{\mathbf{p}}}(\theta_H | i, \hat e_h+ \epsilon) = 1$, and thus $\omega^*_{\hat{\mathbf{p}}}(i, \hat e_h+ \epsilon) = \theta_H$ for every $\epsilon>0$. However, this implies high types would strictly prefer $(i,\hat e_h+ \epsilon)$ over $(j,\hat e_h)$ for all $j\neq i$ since, for $\epsilon,\gamma_{\epsilon}>0$ sufficiently small, 
\[\theta_H - c(\theta_H, \hat e_h+ \epsilon) - \gamma_{\epsilon} >  \hat w_h - c(\theta_H, \hat e_h) = U^{\mathcal{E}^*_{\hat{\mathbf{p}}}} (\theta_H), \]
where the inequality follows from $\hat w_h < \theta_H$. This contradicts the assumption that no students enroll in school $i$ under $\mathcal{E}^*_{\hat{\mathbf{p}}}$. 
Thus, school $i$ can profitably deviate from $p^*$ to $\hat{p} = (\hat{f}, \hat{M}_i)$, contradicting the existence of a pooling RPBE.

\paragraph{Step 5:} \emph{In every semi-pooling RPBE $\mathcal{E}^*$ with zero fees, all students enroll, low types exert $e_l \geq 0$, high types exert $e_l$ with probability $q_h \in (0,1)$ and the Riley effort $e^{R}$ with probability $1-q_h$.} \quad In any semi-pooling symmetric RPBE $\mathcal{E}^*$ with $f^*=0$, schools' monitoring structure must be:
\begin{align}\label{eq:semi-pooling}
    {M}^*(e) =  
   \begin{cases}
        m_A \quad &\text{if } e \geq e_h,\\
        m_B \quad &\text{if } e \in [e_l, e_h),\\
        m_C \quad &\text{otherwise,}
    \end{cases}
\end{align} 
for some $e_h > e_l \geq 0$. Let $\psi^*_{{\mathbf{p}^*}, \theta_L}(i, e_l) =q_l/n$ and $\psi^*_{{\mathbf{p}^*}, \theta_L}(i, 0) =(1-q_l)/n$ for all $i\in I$, $q_l \in (0,1]$. Let $\psi^*_{{\mathbf{p}^*}, \theta_H}(i, e_l) =q_h/n$ and $\psi^*_{{\mathbf{p}^*}, \theta_L}(i, e_h) =(1-q_h)/n$ for all $i\in I$, $q_h \in (0,1)$.\footnote{$q_l = 0$ or $q_h = 0$ corresponds to the case of full separation, which has been characterized by Step 3; $q_h = 1$ corresponds to the case of full pooling, which we have shown is impossible by Step 4.} Denote $w_l=w_l^*=\omega_{\mathbf{p}^*}^*(i,m_B) \in [\max\{\theta_L,0\}, \theta_H)$  for any $i \in I$. 
In this step, we need to show that (i) low types' payoff $U^{\mathcal{E}^*_{\mathbf{p}^*}} (\theta_L)>\max\{\theta_L,0\}$, implying $q_l = 1$; (ii) the threshold effort $e_h = e^{R}$. 

 Suppose, for contradiction, $U^{\mathcal{E}^*_{\mathbf{p}^*}} (\theta_L)\leq \max\{\theta_L,0\}.$ 
Since $\psi^*_{{\mathbf{p}^*}, \theta_L}(i, e_l) =q_l/n>0$, by low types participation constraints, it must be that: 
\begin{align}\label{eqn:payoff_low_semi_zerofee}
    U^{\mathcal{E}^*_{\mathbf{p}^*}} (\theta_L) = w_l - c(\theta_L, e_l) = \max\{\theta_L,0\}.
\end{align}

High types are indifferent between $e_l$ and $e_h$; thus their payoff is
\begin{align}\label{eqn:payoff_high_semi_zerofee}
    U^{\mathcal{E}^*_{\mathbf{p}^*}}(\theta_H) = \theta_H - c(\theta_H, e_h) =
    w_l - c(\theta_H, e_l) = \max\{\theta_L,0\} + c(\theta_L, e_l) - c(\theta_H, e_l),
\end{align}
where the last equality follows from \eqref{eqn:payoff_low_semi_zerofee}.

Consider the subgame following  $\tilde{\mathbf{p}}$ where $\tilde p_j=p_j^*$ for all $j\neq i$ and school $i$ unilaterally deviates to $\tilde p_i = (\tilde{f}_i, \tilde{M}_i)$ such that  $\tilde{f}_i = \gamma_{\epsilon} >0$, and $\tilde{M}_i$ is given by \eqref{eq:M tilde}. We will show that school $i$ can profitably attract a positive mass of students with this unilateral deviation.

Suppose, for contradiction, no students enroll in school $i$ in $\mathcal{E}^*_{\tilde{ \mathbf{p}}}$, i.e., $\psi^*_{{\tilde{\mathbf{p}}}, \theta_H}(i)=\psi^*_{\tilde{\mathbf{p}}, \theta_L}(i)=0$. 
Denote by $\tilde e_h \in \{0, e_l, e_h\}$ the maximal effort high types exert outside school $i$ in
$\mathcal{E}^*_{\tilde{ \mathbf{p}}}$. 

If $\tilde e_h = e_h$, we have
\begin{align}\label{ineq:D1_zerofee}
    U^{\mathcal{E}_{\tilde{\mathbf{p}}}}(\theta_H) - U^{\mathcal{E}_{\tilde{\mathbf{p}}}}(\theta_L) &\leq \theta_H - c(\theta_H, e_h) - \max\{\theta_L,0\} \nonumber\\
    &= c(\theta_L, e_l)- c(\theta_H, e_l) < c(\theta_L, e_l+\epsilon)- c(\theta_H, e_l+\epsilon),
\end{align}
where the equality follows from \eqref{eqn:payoff_high_semi_zerofee}. Since no students enroll in school $i$, $(i,m_A)$ is unsent under $\mathcal{E}^*_{\tilde{\mathbf{p}}}$. Thus, our refinement and inequality \eqref{ineq:D1_zerofee} imply 
$\mu^*_{\tilde{\mathbf{p}}}(\theta_H | i,m_A) = 1$, and thus $\omega^*_{\tilde{\mathbf{p}}}(i,m_A) = \theta_H$ for every $\epsilon>0$. However, this implies high types would strictly prefer $(i,e_l+ \epsilon)$ over $(j, e_h)$ for all $j\neq i$ since, for $\epsilon,\gamma_{\epsilon}>0$ sufficiently small, 
\begin{align}\label{ineq:high_compare_3}
    \theta_H - \gamma_\epsilon - c(\theta_H, e_l+ \epsilon) > \theta_H - c(\theta_H, e_h) \geq U^{\mathcal{E}_{\tilde{\mathbf{p}}}} (\theta_H),
\end{align}
contradicting $\psi^*_{{\tilde{\mathbf{p}}}, \theta_H}(i)=\psi^*_{\tilde{\mathbf{p}}, \theta_L}(i)=0$.

If $\tilde e_h \in \{0, e_l\}$, similar arguments as in Step 4 prove that school $i$ can attract all high types: they would strictly prefer $(i,\tilde e_h + \epsilon')$ over $(j,\tilde e_h )$ when $\epsilon',\gamma_{\epsilon} > 0$ are sufficiently small.
Therefore, with the unilateral deviation to $\tilde p_i$, school $i$ can earn a strictly positive profit, i.e., more than the 0 RPBE profit it would obtain by choosing $p^*$.

Therefore, it must be that $U^{\mathcal{E}^*_{\mathbf{p}^*}} (\theta_L) > \max\{\theta_L,0\}$, and thus $\psi^*_{{{\mathbf{p}^*}}, \theta_H}(i)=\psi^*_{{\mathbf{p}^*}, \theta_L}(i)=1/n$ for all $i\in I$.

Next, we prove that $e_h = e^{R}$. To this end, consider
school $i$'s unilateral deviation to $\hat{p}_i = (\hat{f}_i, \hat{M}_i)$ where  $\hat{f}_i = \gamma_{\epsilon} >0$, and $\hat{M}_i(e)=e$ for all $e\geq 0$. Denote by $\mathcal{E}^*_{\hat{\mathbf{p}}}$ the EPBE prescribed by $\mathcal{E}^*$ following $\hat{\mathbf{p}}$, where $\hat{p}_j=p^*$ for all $j\in I\setminus \{i\}$ and $\hat{p}_i = (\hat{f}_i, \hat{M}_i)$. Denote by $\hat e_h \in \{0, e_l, e_h\}$ the minimal  effort high types exert outside $i$, and  define $\hat w_h:=\max_{j\neq i}\omega_{{\hat{\mathbf{p}}}}(j,\hat e_h) \geq 0$ the maximal wage associated with $\hat e_h$ at a school $j \neq i$ under $\mathcal{E}^*_{\hat{\mathbf{p}}}$. 
We show that, whenever $e_h \neq e^{R}$,
school $i$'s unilateral deviation to  $\hat{p}_i$ is strictly profitable.

First, suppose that $e_h < e^{R}$. Since $\hat e_h \leq e_h < e^{R}$, we have $\theta_h-f^*-c(\theta_L, \hat e_h)>\max\{0,\theta_L-f^*\}$, implying low types would optimally pool with high types on $\hat e_h$  at school $j\neq i$  to earn $\hat w_h$, if $\hat w_h\geq \theta_H$. This implies $\hat w_h< \theta_H$. Then, the same argument as in Step 4 proves that by unilaterally deviating to $\hat{p}_i$, school $i$ can attract strictly positive enrollment and earn profits $\hat{ \pi}_i>0=\pi^* = 0$. As a result, we must have that $e_h \geq e^R$.

Second, suppose $e_h > e^{R}$. If $\hat e_h<e_h$ the same argument as the previous paragraph implies school $i$ earns strictly positive profits in $\mathcal{E}^*_{\hat{\mathbf{p}}}$. If instead $\hat e_h=e_h$, same argument as in Step 3 proves that school $i$ would attract some high types in $\mathcal{E}^*_{\hat{\mathbf{p}}}$: they would strictly prefer $(i,e^R)$ over  $(j, e_h )$, with $j\neq i$, when $\gamma_{\epsilon} > 0$ is sufficiently small. Thus, school $i$ would profitably deviate from $p^*$, yielding $\pi_i^*=0$, to $\hat{p}_i$, yielding strictly positive profits--a contradiction.

We conclude that $e_h = e^{R}$.

\paragraph{Step 6:} \emph{Under IIS, any RPBE with $f^*=0$ induces the Riley outcome.} \quad

So far, we have established that (i) any separating RPBE delivers the Riley outcome (Step 3); (ii) no RPBE features full pooling (Step 4); (iii) a semi-pooling RPBE with $f^*=0$ is characterized as in Step 5. This step shows how the IIS assumption rules out the semi-pooling RPBE with zero fees.  

Suppose, for contradiction, an IIS-RPBE $\mathcal{E}^*$ is characterized by Step 5: schools charge $f^*= 0$ and adopt monitoring structures as in \eqref{eq:semi-pooling}  with $e_h = e^{R} > e_l \geq 0$; (ii) all students enroll; low types exert $e_l$; high types exert effort $e_l$ with probability $q_h \in (0,1)$ and $e^{R}$ with probability $1-q_h$; (iii) firms offer wages $\theta_H, w_l, \max\{\theta_L, 0\}$ upon receiving message $m_A, m_B, m_C$, respectively, from any school, where $w_l$ is given by \eqref{eq:mixed wage}.

Consider $\mathcal{E}^*_{\tilde {\mathbf{p}}}$, the EPBE prescribed by $\mathcal{E}^*$ in the subgame following $\tilde {\mathbf{p}}$,  with $\tilde{p}_j = p^*_j$ for all $j \neq i$ and $\tilde{p}_{i}=(\tilde{f}_i, \tilde{M}_i)$ where $\tilde{f}_i = \gamma_{\epsilon}>0$ sufficiently small, and $\tilde{M}_i$ is given by \eqref{eq:M tilde}. 

Under IIS, if no student enrolls in school $i$ under $\mathcal{E}^*_{\tilde {\mathbf{p}}}$, then all students enroll in other schools, adopt the same conditional strategies $\psi^*_{\tilde{\mathbf{p}},\theta}(j,e|j \neq i)=\psi^*_{\mathbf{p}^*,\theta}(j,e|j \neq i)$, and earn the same payoffs as in  $\mathcal{E}^*_{{\mathbf{p}^*}}$. Thus
\begin{align*}
    U^{\mathcal{E}^*_{\tilde{\mathbf{p}}}}(\theta_H) - U^{\mathcal{E}^*_{\tilde{\mathbf{p}}}}(\theta_L) &= U^{\mathcal{E}^*_{{\mathbf{p}^*}}}(\theta_H) - U^{\mathcal{E}^*_{{\mathbf{p}^*}}}(\theta_L)\\
    &= c(\theta_L, e_l)- c(\theta_H, e_l) < c(\theta_L, e_l+\epsilon)- c(\theta_H, e_l+\epsilon).
\end{align*}
Thus, our refinement implies 
$\mu^*_{\tilde{\mathbf{p}}}(\theta_H | i,m_A) = 1$, and thus $\omega^*_{\tilde{\mathbf{p}}}(i,m_A) = \theta_H$ for every $\epsilon>0$. However, this implies high types would strictly prefer $(i,e_l+ \epsilon)$ over $(j, e_h)$ for all $j\neq i$ since, for $\epsilon,\gamma_{\epsilon}>0$ sufficiently small, 
\begin{align}
    \theta_H - \gamma_\epsilon - c(\theta_H, e_l+ \epsilon) > \theta_H - c(\theta_H, e_h) = U^{\mathcal{E}^*_{\tilde{\mathbf{p}}}} (\theta_H),
\end{align}
contradicting $\psi^*_{{\tilde{\mathbf{p}}}, \theta_H}(i)=\psi^*_{\tilde{\mathbf{p}}, \theta_L}(i)=0$.
 Hence, school $i$ can attract a positive mass of students and thus earn a positive profit $\tilde{\pi}_i>\pi^*=0$, a contradiction. Therefore, there exists no semi-pooling IIS-RPBE with $f^*=0$.

\section{Proof of Proposition \ref{prop:mon_screenConst}}

Observe that, in the sorting case (i.e., $\theta_L > 0$) with $K \geq \E\theta$, the unique RPBE outcome is described as in Proposition \ref{prop:mon_sorting} following the logic in Section \ref{sec:mon_sorting}. In what follows, we analyze the screening case (i.e., $\theta_L < 0$) and the sorting case with $K < \E\theta$.
   
\paragraph{Step 1}: 
\emph{The school cannot make profits higher than 
$\pi^*= K[\lambda + \tilde{\alpha}_K(1-\lambda)]$
where $\tilde{\alpha}_K = \min\{\alpha_K, 1\}$ and $\alpha_K$ is given by \eqref{eq:low enroll prob}
.} \quad 

First fix any fee $f \in [0,\theta_H]$. Define $\tilde{\alpha}_f=\min \{{\alpha}_f, 1\}$ such that 
\[\frac{\lambda \theta_H+{\alpha}_f(1-\lambda)\theta_L}{\lambda +{\alpha}_f(1-\lambda)}=f.\]
Thus, $\tilde{\alpha}_f  \in [0,1]$ represents the largest proportion of low types that can optimally enroll under the fee $f$. Accordingly, the school's profit from charging the fee $f$ is at most 
\begin{align*}
    \bar{\pi}_f &= f[\lambda+\tilde{\alpha}_f(1-\lambda)]\\
    &= \begin{cases}
        f &\text{ if } f \in [0, \E\theta),\\
        \lambda \theta_H+ {\alpha}_f(1-\lambda)\theta_L &\text{ if } f \in [\E\theta, \theta_H].
    \end{cases}
\end{align*}
With the credit constraint $K$, the fee charged by the school must satisfy $f \in [0, K]$. In the screening case $(\theta_L < 0)$, since $\bar{\pi}_f$ increases with $f$, the school's profit cannot exceed $\pi^*= \bar{\pi}_K = K [\lambda +\tilde{\alpha}_K(1-\lambda)]$. In the sorting case $(\theta_L > 0)$ with $K \in (0, \E\theta)$, since $\bar{\pi}_f$ increases with $f$ over $[0,K]$, the school's profit cannot exceed $\pi^*= \bar{\pi}_K = K$.

\paragraph{Step 2}: \emph{The candidate PBE with on-path play described by Proposition \ref{prop:mon_screenConst} and an EPBE in every off-path subgame is an RPBE that achieves profits $\pi^*$.} \quad The equilibrium outcome specified in Proposition \ref{prop:mon_screenConst} yields profit $\pi^*$ for the school.
The school will not deviate as it can never obtain a profit higher than $\pi^*$. If $K\in [\E\theta, \theta_H)$, then students will be indifferent between enrolling or not, obtaining zero payoffs either way; if $K\in [0, \E\theta)$, then students weakly prefer to enroll (see Step 4). Firms' wages are given following the Bayes' rule. Finally, since there is no unsent message under the equilibrium policy, the proposed strategy profile is consistent with our refinement.

\paragraph{Step 3}: \emph{No other equilibrium profits.} 
 Given Step 1, we only need to show the school's equilibrium profit cannot be lower than $\pi^*$. Suppose for  contradiction that there exists an RPBE $(f', M', \psi', \omega',\mu')$ in which the school earns $\pi' <\pi^*$. Consider the following deviation from the school: 
\begin{align*}
    &\hat f = K, \\
    &\hat M (e) =
    \begin{cases}
         e \quad &\text{if } e < \epsilon,\\
         \epsilon \quad &\text{if } e\geq   \epsilon,
    \end{cases}
\end{align*}
where $\epsilon>0$ is arbitrarily small.

We show below that under this policy $\hat p = (\hat f, \hat M)$, the schools profit $\hat{\pi}$ converges to $\pi^*$ as $\epsilon$ converges to $0$.

To see this, note first that a positive mass of high types must enroll in the school: if no high-type student enrolls, then the message $\epsilon$ must be unsent, and all students must obtain zero payoff under the policy $\hat p$. Our refinement then implies that exerting effort $\epsilon$ is rewarded by wage $\theta_H$. Since $\hat f=K<\theta_H$, high types would prefer to enroll and exert a small effort $\epsilon$, a contradiction. 

Denote by $\hat{\underline{e}}_h \in [0, \epsilon]$ the minimal effort that high types exert in school under the policy $\hat p$. 
With $\epsilon$ sufficiently small, some low types must also exert effort $\hat{\underline{e}}_h$, and thus the associated wage $\omega(\hat{{p}},\hat{\underline{e}}_h)<\theta_H$.

If $\hat{\underline{e}}_h <\epsilon$, then
by our refinement, $\omega(\hat{{p}},\gamma)=\theta_H$ for any message $\gamma\in (\hat{\underline{e}}_h,\epsilon]$, attracting students to deviate and exert a negligible higher effort, a contradiction. Therefore, we have (i) $\hat{\underline{e}}_h =\epsilon$, (ii) all high types and a positive mass of low types enroll in the school, and (iii) all enrolled students exert effort $\epsilon$.   
Finally, we show how low types' enrollment varies with the magnitude of the credit constraint $K$. If $\E\theta \geq K +c(\theta_L, \epsilon) $, then all low types enroll, and the school's profit equals $K$. If $\E\theta < K+c(\theta_L, \epsilon)$, then the fraction of low types that enroll in the school, $\hat{\alpha}_\epsilon$, is given by $\frac{\lambda \theta_H+\hat{\alpha}_\epsilon(1-\lambda)\theta_L}{\lambda +\hat{\alpha}_\epsilon(1-\lambda)}=K+c(\theta_L, \epsilon)$. Note that $\lim_{\epsilon \rightarrow 0} \hat{\alpha}_\epsilon = {\alpha}_K$, where ${\alpha}_K$ is given by \eqref{eq:low enroll prob}.  Thus the school's profit $\hat \pi =K(\lambda +\min\{\hat{\alpha}_\epsilon,1\}(1-\lambda))$ converge to $\pi^* > \pi'$ as $\epsilon$ goes to 0, a profitable deviation and therefore a contradiction. Hence, the monopolist school's equilibrium profit must be $\pi^*.$ 

\paragraph{Step 4}: \emph{If $K \in [ \E\theta, \theta_H)$, then the RPBE outcome is unique; if $K\in (0,\E\theta)$, there exist multiple RPBE outcomes, all featuring full enrollment.}   \quad 
By Step 3, the school's equilibrium profit is $\pi^* = K[\lambda + \tilde{\alpha}_K(1-\lambda)]$. For $K \in [\E\theta, \theta_H)$, we focus on the screening case and we have $\tilde{\alpha}_K = {\alpha}_K$. By Step 1, $\pi^*$ can only be attained by (i) the school charging the fee $K$; (ii) all enrolled students exerting zero effort (and thus by minimality requirement, the school must adopt a pooling monitoring policy); (iii) all high types and a fraction ${\alpha}_K$ of low types enrolling in the school; and (iv) firms offer wage $K$ to all students enrolled in the school. We have shown in Step 2 that this outcome can arise in an RPBE; the arguments above further show that this is the unique RPBE outcome that can arise.

If $K\in (0,\E\theta)$, then $\tilde{\alpha}_K = 1$ (i.e., full enrollment) and $\pi^* = K$ (and thus the fee must be $K$). For the screening case,
using the arguments similar to Step 2, we can show that there exists a continuum of full pooling RPBE outcomes where (i) the school adopts the following policy
\begin{align*}
    &f^* = K, \\
    &M^*(e) =
    \begin{cases}
         m_A \quad &\text{if } e \geq e_l,\\
         m_B \quad &\text{if } e < e_l,
    \end{cases}
\end{align*}
where $e_l\leq [0, e']$ where $ K + c(\theta_L, e') = \E\theta$; (ii)
all students enroll and exert the threshold effort $e_l$; (iii) firms offer the wage $\E\theta$ upon receiving the message $m_A$, and offer zero wage upon receiving the message  $m_B$.  Additionally, there exists a continuum of partially pooling RPBE outcomes where (i) the school adopts the following policy
\begin{align*}
    &f^* = K, \\
    &M^*(e) =
    \begin{cases}
         m_A \quad &\text{if } e \geq e_h,\\
         m_B \quad &\text{if } e \in [e_l, e_h),\\
         m_C \quad &\text{if } e \in [0, e_l),
    \end{cases}
\end{align*}
where $e_l \leq [0, e_h)$; (ii) low types enroll and exert effort $e_l$;  (iii) high types enroll, exert effort $e_l$ with probability $q_h \in (0,1)$ and exert effort $e_h$ with probability $1- q_h$ such that 
\[w_l - c(\theta_H, e_l) = \theta_H-c(\theta_H, e_h),\]
where $w_l = \frac{\lambda q_h\theta_H+(1-\lambda)\theta_L}{\lambda q_h+(1-\lambda)
} \geq K+c(\theta_L, e_l)$; (iv) firms offer the wage $\theta_H$ upon receiving the message $m_A$, the wage $w_l$ upon receiving the message  $m_B$, and the wage zero upon receiving the message  $m_C$.

Similarly, for the sorting case with $K\in (0,\E\theta)$, we can show that there exists a continuum of RPBE outcomes where (i) the school adopts the following policy
\begin{align*}
    &f^* = K, \\
    &M^*(e) =
    \begin{cases}
         m_A \quad &\text{if } e \geq e_l,\\
         m_B \quad &\text{if } e < e_l,
    \end{cases}
\end{align*}
where $e_l\geq 0$; (ii) low types enroll and exert zero effort;  (iii) high types enroll, exert zero effort with probability $q_h \in [0,1]$ and exert effort $e_l$ with probability $1- q_h$ such that 
\[w_l = \theta_H-c(\theta_H, e_l),\]
where $w_l = \frac{\lambda q_h\theta_H+(1-\lambda)\theta_L}{\lambda q_h+(1-\lambda)
} \geq K$; (iv) firms offer the wage $\theta_H$ upon receiving the message $m_A$, the wage $w_l$ upon receiving the message  $m_B$.

\end{document}